\documentclass[lettersize,journal]{IEEEtran}
\newcommand{\bl}{\textcolor{black}}

\usepackage{amsmath,amsfonts,amssymb,amsthm}
\usepackage{algorithmic}
\usepackage{array}
\usepackage{subcaption}
\usepackage{graphicx}
\usepackage{multirow}
\usepackage{textcomp}
\usepackage{stfloats}
\usepackage{url}
\usepackage{verbatim}
\usepackage{graphicx}
\usepackage{color, soul}
\usepackage{bm}
\usepackage{hyperref}

\usepackage{cite}
\usepackage{xcolor}
\usepackage{xpatch}

\newtheorem{theorem}{Theorem}

\newtheorem{lemma}{\textbf{Lemma}}

\usepackage{algorithmic}
\usepackage[linesnumbered,ruled]{algorithm2e}
\usepackage{cuted}

\def\BibTeX{{\rm B\kern-.05em{\sc i\kern-.025em b}\kern-.08em
    T\kern-.1667em\lower.7ex\hbox{E}\kern-.125emX}}
\usepackage{balance}

\begin{document}
\title{Digital Twin-based 3D Map Management for Edge-assisted Device Pose Tracking in Mobile AR}

\author{    Conghao~Zhou,~\IEEEmembership{Member,~IEEE,}
            Jie~Gao,~\IEEEmembership{Senior~Member,~IEEE,}
            Mushu~Li,~\IEEEmembership{Member,~IEEE,}
            Nan~Cheng,~\IEEEmembership{Senior Member,~IEEE,}
            Xuemin~(Sherman)~Shen,~\IEEEmembership{Fellow,~IEEE,}
            and Weihua~Zhuang,~\IEEEmembership{Fellow,~IEEE}

            \thanks{This work was supported by a research grant from the Natural Sciences and Engineering Research Council (NSERC) of Canada.}

            \thanks{C.~Zhou, X.~Shen, and W.~Zhuang are with the Department of Electrical and Computer Engineering, University of Waterloo, Waterloo, ON, N2L 3G1, Canada, (e-mail:\{c89zhou, sshen, wzhuang\}@uwaterloo.ca).}

            \thanks{J.~Gao is with the School of Information Technology, Carleton University, Ottawa, ON, Canada K1S 5B6 (email:~jie.gao6@carleton.ca).}
            \thanks{M.~Li is with the Department of Electrical, Computer, and Biomedical Engineering, Toronto Metropolitan University, Toronto, ON, M5B 2K3, Canada (e-mail:~mushu1.li@torontomu.ca).}
            \thanks{N. Cheng is with the School of Telecommunications Engineering, Xidian University, Xi'an, 710071, P.R.China (e-mail:~nancheng@xidian.edu.cn).}
            \thanks{Part of this work was presented at IEEE/CIC ICCC 2023 \cite{zhou2023digital}.}
        }

\maketitle

\begin{abstract}

Edge-device collaboration has the potential to facilitate compute-intensive device pose tracking for resource-constrained mobile augmented reality (MAR) devices. In this paper, we devise a 3D map management scheme for edge-assisted MAR, wherein an edge server constructs and updates a 3D map of the physical environment by using the camera frames uploaded from an MAR device, to support local device pose tracking. Our objective is to minimize the uncertainty of device pose tracking by periodically selecting a proper set of uploaded camera frames and updating the 3D map. To cope with the dynamics of the uplink data rate and the user's pose, we formulate a Bayes-adaptive Markov decision process problem and propose a digital twin (DT)-based approach to solve the problem. First, a DT is designed as a data model to capture the time-varying uplink data rate, thereby supporting 3D map management. Second, utilizing extensive generated data provided by the DT, a model-based reinforcement learning algorithm is developed to manage the 3D map while adapting to these dynamics. Numerical results demonstrate that the designed DT outperforms Markov models in accurately capturing the time-varying uplink data rate, and our devised DT-based 3D map management scheme surpasses benchmark schemes in reducing device pose tracking uncertainty.

\end{abstract}


\section{Introduction}

In the sixth-generation (6G) networks, immersive communications are anticipated to transcend the existing communication paradigm by offering users highly realistic and interactive experiences~\cite{shen2021holistic}. Augmented reality (AR), as a representative form of immersive communications, aims to seamlessly integrate virtual objects into the surrounding physical environments users, thereby enabling them to interact with virtual objects in a lifelike manner~\cite{zhang2023location}. Despite decades of development, AR has not been adopted in our daily lives on a large scale due to limitations such as device size~\cite{yang2019multi,zhang2022edge}. With rapid advancement in mobile devices, including smartphones and smart glasses, mobile AR (MAR) technology is expected to penetrate various fields in the 6G era, unlocking opportunities for a wide range of applications, such as immersive learning and tourism~\cite{shen2023toward,tan2020uav}. 

Tracking the time-varying pose of each MAR device is indispensable for MAR applications. Generally, to geometrically align the virtual objects with the physical environment within the field of view (FoV) of each MAR device in a 3D manner, the spatial relationship between the MAR device and the physical environment needs to be determined~\cite{ran2019sharear}. Nowadays, the real-time information on the required 3D spatial relationship can be provided by the simultaneous localization and mapping (SLAM) technique, which can be used to estimate the 3D position and orientation, jointly referred to as the \emph{3D device pose}, of an MAR device relative to the physical environment within its FoV~\cite{piao2019real}. As a result, SLAM-based 3D device pose tracking\footnote{``Device pose tracking'' is also called ``device localization'' in some works.} is anticipated to be a common module used by emerging MAR platforms, e.g.,~ARKit~\cite{linowes2017augmented} and ILLIXR~\cite{huzaifa2021illixr}, for supporting the development of various MAR applications. Despite the capability of SLAM in 3D alignment for MAR applications, limited resources hinder the widespread implementation of SLAM-based 3D device pose tracking on MAR devices. The primary limitation arises from the excessive resources demanded by SLAM or its variant techniques, beyond what are typically available on MAR devices~\cite{ben2022edge}. Specifically, to achieve accurate 3D device pose tracking, SLAM techniques need the support of a 3D map that consists of a large number of distinguishable landmarks in the physical environment. However, obtaining and maintaining such a 3D map for continuously updating previously device poses consume excessive storage and computing resources~\cite{siriwardhana2021survey}.

\bl{Cloud/edge-assisted} device pose tracking offers a promising solution to address the resource limitations of MAR devices by leveraging network resources~\cite{apicharttrisorn2020characterization,han2022intelli}. From cloud-computing-assisted tracking to the recently prevalent mobile-edge-computing-assisted tracking, researchers have explored resource-efficient approaches for network-assisted tracking from different perspectives. Research works in one category are from the perspective of device pose tracking primarily, which focus on refining SLAM system design to facilitate cloud/edge-device collaboration for MAR~\cite{ran2019sharear,chen2018marvel}. However, these research works tend to overlook the impact of network dynamics by assuming time-invariant communication resource availability or delay constraints. Meanwhile, studies in another category have delved into cloud/edge computing task offloading and scheduling from a networking perspective, considering dynamic service demand and resource availability~\cite{pan2023joint,du2023gradient,zhou2022digital}. Treating device pose tracking as a computing task, these approaches are apt to optimize networking-related performance metrics such as delay but do not capture the impact of computing task offloading and scheduling on the performance of device pose tracking. Consequently, despite considerable research efforts from both perspectives, network-assisted device pose tracking that \emph{natively} adapts to network dynamics with optimal device pose tracking performance remains a significant challenge for MAR. 

To fill the gap between the aforementioned two categories of research works, we investigate \emph{network dynamics-aware} 3D map management for network-assisted tracking in MAR. Specifically, we consider an edge-assisted SALM architecture, in which an MAR device conducts real-time device pose tracking locally and uploads the captured camera frames to an edge server. The edge server constructs and updates a 3D map using the uploaded camera frames to support the local device pose tracking. We optimize the performance of device pose tracking in MAR by managing the 3D map, which involves uploading camera frames and updating the 3D map. There are three key challenges to 3D map management \bl{for individual MAR devices}. First, an MAR device must select only a portion of the collected information, \bl{more specifically camera frames}, on its physical environment to update the 3D map for its 3D device pose tracking, given the computing and storage resource constraints at the edge server~\cite{campos2021orb}. Second, the camera frame uploading at an MAR device must adapt to the time-varying uplink data rate of the MAR device, which determines the maximum number of camera frames that can be uploaded per unit time for 3D map update~\cite{chen2023adaptslam}. Third, a new \bl{performance} metric different from tracking accuracy for evaluating the device pose tracking performance becomes necessary when the network perspective is integrated into 3D map management, due to the lack of ground truth for the real-time 3D pose of an MAR device in practice~\cite{huzaifa2021illixr}.

To address these challenges, we introduce a digital twin (DT)-based approach to effectively cope with the dynamics of the uplink data rate and the device pose. Building upon the DT architecture delineated in our previous work~\cite{shen2021holistic}, we establish a DT for an MAR device to create a data model that can infer the unknown dynamics of its uplink data rate. Subsequently, we propose an artificial intelligence (AI)-based method, which utilizes the data model provided by the DT to learn the optimal policy for 3D map management in the presence of device pose variations. The main contributions of this paper are as follows:
    \begin{itemize}
        \item We introduce a new performance metric, termed pose estimation uncertainty, to indicate the long-term impact of 3D map management on the performance of device pose tracking, which adapts conventional device pose tracking \bl{in MAR} to network dynamics.

        \item We establish a user DT (UDT), which leverages deep variational inference to extract the latent features underlying the dynamic uplink data rate. The UDT provides these latent features to simplify 3D map management and support the emulation of the 3D map management policy in different network environments.

        \item We develop an adaptive and data-efficient 3D map management algorithm featuring model-based reinforcement learning (MBRL). By leveraging the combination of real data from actual 3D map management and emulated data from the UDT, the algorithm can provide an adaptive 3D map management policy in highly dynamic network environments.

    \end{itemize}

The remainder of this paper is organized as follows. Section~II provides an overview of related works. Section~III describes the considered scenario and system models. Section~IV presents the problem formulation and transformation. Section~V introduces our UDT, followed by the proposed MBRL algorithm based on the UDT in Section~VI. Section~VII presents the simulation results, and Section~VIII concludes the paper.

\section{Related Works}

In this section, we first summarize existing works on edge/cloud-assisted device pose tracking from the MAR or SLAM system design perspective. Then, we present some related works on computing task offloading and scheduling from the networking perspective.

\subsection{Cloud/Edge-assisted Device Pose Tracking}

Existing studies on edge/cloud-assisted MAR applications can be classified based on their approaches to aligning virtual objects with physical environments. Specifically, there are image retrieval-based, deep learning-based, and localization-based approaches~\cite{chen2018marvel}.

The image retrieval-based approaches utilize a pre-constructed database comprising labeled images, deployed at a cloud/edge server~\cite{zhang2022sear}. Given a captured camera frame, an MAR device searches and retrieves the most similar labeled image from the database. Subsequently, the information from this retrieved labeled image is utilized to support the 3D alignment of virtual objects with this captured camera frame. Deep learning-based approaches in MAR can be viewed as an advancement over image retrieval-based approaches. To overcome the low efficiency of image retrieval-based approaches, deep learning-based approaches leverage deep neural networks (DNNs), e.g., convolutional neural networks, to find the most similar labeled image~\cite{liu2020collabar}. Both image retrieval-based and deep learning-based approaches are only suitable for lightweight MAR applications that do not need large databases~\cite{chen2018marvel}. Since a physical object can be viewed from various angles and distances, \bl{these} approaches require a large \bl{set of} distinct labels. In addition, the accuracy of both approaches in 3D alignment is limited for existing MAR applications~\cite{ran2020multi}. 

Currently, both industries and academia have shifted their focus towards localization-based approaches, e.g., Visual-SLAM~\cite{huzaifa2021illixr}. By establishing 3D maps for the physical environments, localization-based approaches can estimate the 3D poses of individual MAR devices with high accuracy. Instead of identifying physical objects based on their appearance, localization-based approaches can leverage location-related information of physical objects to facilitate accurate and resource-efficient 3D alignment. Chen~\emph{et~al.} utilize a cloud server to calibrate the localization of a local MAR device~\cite{chen2018marvel}. The MAR device uploads recent camera frames when there is a significant discrepancy between the localization result from the cloud server and that from the MAR device. Ben Ali~\emph{et~al.} build an edge-device collaboration system for Visual-SLAM, with 3D map management on the edge server and 3D pose estimation on the local MAR device. Following~\cite{ben2022edge}, the authors of~\cite{chen2023adaptslam} investigate the impact of radio resource constraints and introduce pose estimation uncertainty in edge-assisted Visual-SLAM. Extending edge-device collaboration to support multiple MAR users, the works in~\cite{ran2020multi,dhakal2022slam} focus on the coordinate synchronization to guarantee spatial consistency across different MAR devices. Ren~\emph{et al.} investigate the computing and communication resource allocation to support coordinate synchronization~\cite{ren2020edge}. Despite the existing efforts towards SLAM system design in cloud/edge-assisted MAR, the \bl{impact} of network dynamics on device pose tracking performance remains open.

We employ a Visual-SALM technique, as a localization-based approach, to enhance edge-assisted 3D pose tracking in MAR. Different from conventional localization-based approaches that often overlook network dynamics and assume 3D maps of unlimited size, we emphasize the long-term impact of network dynamics on 3D map management and propose a DT-based approach to adapt to the dynamics of the uplink data rate and the user's pose, while considering a limited-size 3D map given the resource constraints at the edge server.

\subsection{Computing Task Offloading and Scheduling}

By treating the tracking of the device pose for each camera frame as a computing task, the process of uploading camera frames and updating a 3D map is closely related to the computing task offloading and scheduling in networking~\cite{shen2023toward}. Depending on the chosen performance metrics, existing approaches to computing task offloading and scheduling differ significantly.

Many studies concentrate on improving the delay performance of computing task offloading in a specific network scenario, including space-air-ground integrated networks~\cite{cheng2019space,ji2023cooperative} and vehicular networks~\cite{hui2022collaboration}. Meanwhile, some researchers investigate computing task offloading or scheduling schemes for specific applications. Li~\emph{et al.} focus on virtual reality applications and aim to reduce the camera frame missing rate in dynamic network environments~\cite{li2023user}. Considering surveillance and search-and-rescue-related applications with unmanned aerial vehicles (UAVs), the authors in~\cite{luo2022deep} propose to improve the reliability of target search results by properly offloading the search tasks of UAVs according to UAV trajectories. With the advent of artificial intelligence (AI)-related applications, researchers have started to investigate computing task offloading or scheduling strategies to optimize the accuracy of AI-related applications. The studies in~\cite{pan2023joint} and~\cite{hu2023adaptive} focus on the inference accuracy of DNNs utilized for AR and Internet of Things, respectively. To facilitate federated learning, Du~\emph{et al.} propose a task scheduling scheme for distributed devices according to their data qualities and channel conditions~\cite{du2023gradient}.

Different from the aforementioned works on computing task offloading and scheduling, our approach incorporates device pose estimation uncertainty as a performance metric to evaluate camera frame uploading and 3D map update in MAR applications. Furthermore, we prioritize camera frames when updating the 3D map to accommodate user pose variations given the time-varying uplink data rate.

\section{System Model}

\subsection{Considered Scenario}

    \begin{figure}[t]
        \centering,
        \includegraphics[width=0.5\textwidth]{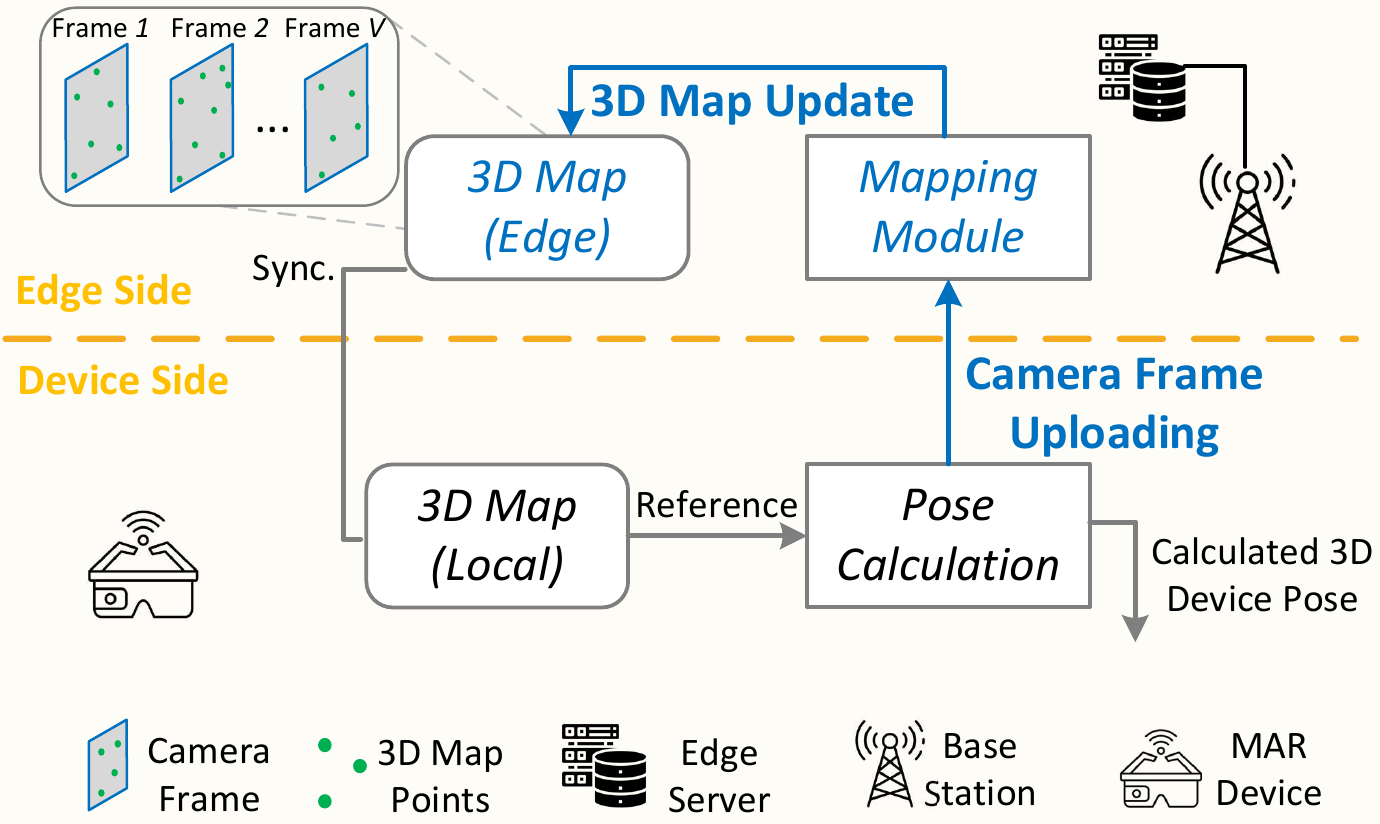}
        \caption{The considered scenario of edge-assisted MAR.}\label{system}
    \end{figure}

Let one MAR device run an MAR application. While the 3D pose (i.e., the position and orientation) of the MAR device changes over time, the physical environment (e.g., a living room) of the MAR device does not change~\cite{cozzolino2022nimbus}. To establish the spatial alignment between virtual objects from the MAR application and the physical environment, the MAR device needs to periodically capture camera frames for tracking its 3D pose as it moves while updating a 3D map of the physical environment. A \emph{3D map} \bl{consists of} a set of captured camera frames and the corresponding set of \bl{3D map points} in these camera frames. \bl{Each 3D map point is referred to as a feature point}, which corresponds to a distinctive spot or characteristic (e.g., a corner of wall) of the physical environment~\cite{linowes2017augmented}.

Generally, the device pose tracking for MAR applications comprises two modules: a lightweight pose calculation module for real-time 3D pose calculation and a resource-intensive mapping module for managing a 3D map of the physical environment~\cite{campos2021orb,mur2017orb}. To calculate the device pose corresponding to a particular camera frame, the feature points detected in this camera frame need to be matched with feature points in previously captured camera frames, which are stored in the 3D map. The 3D map management in the mapping module involves constructing and updating the 3D map as the reference for 3D pose calculation. 

Due to the limited computing capability and battery of the MAR device, we adopt an edge-device collaborative framework, \bl{as} shown in Fig.~\ref{system}, to support the MAR application. Specifically, an edge server at a base station (BS) is equipped with the mapping module for 3D map management, and the MAR device is equipped with the pose calculation module for local 3D pose calculation.

\subsection{Workflow of Edge-assisted Device Pose Tracking}

The general workflow of device pose tracking for edge-device collaborative MAR applications includes four steps: 
    \begin{enumerate}
        \item \emph{3D pose calculation}: The MAR device calculates its 3D pose corresponding to each camera frame by matching the 3D map points (i.e., feature points of the physical environment) detected in this camera frame with those contained in the local 3D map, shown as the ``3D Map (Local)'' block in Fig.~\ref{system};

        \item \emph{Camera frame uploading}: The MAR device uploads a subset of recently captured camera frames to the edge server depending on its available uplink communication resource;

        \item \emph{3D map update}: The edge server updates the 3D map, shown as~the ``3D Map (Edge)'' block in Fig.~\ref{system}, by processing the camera frames uploaded by the MAR device;

        \item \emph{Synchronization}: The edge server periodically sends the updated 3D map back to the MAR device as references to facilitate 3D pose calculation~\cite{ben2022edge}.
    \end{enumerate}

The mapping module at the edge server and the pose calculation module at the MAR device operate on two different time scales. Specifically, the 3D pose calculation (Step~1) is conducted for each camera frame and takes as short as several milliseconds to complete, while the 3D map management (Steps~2-4) generally operates on a larger time scale (e.g., over several seconds)~\cite{campos2021orb}. In this paper, we focus on~\emph{camera frame uploading} (Step 2) and \emph{3D map update} (Step 3) corresponding to the blue arrows in Fig.~\ref{system}, which are detailed in Subsections~\ref{sec22} and~\ref{sec23}, respectively. For brevity, the term ``3D map'' in the rest of the paper denotes ``the 3D map managed by the edge server'' unless otherwise stated.

    \begin{figure}[t]
        \centering
        \includegraphics[width=0.45\textwidth]{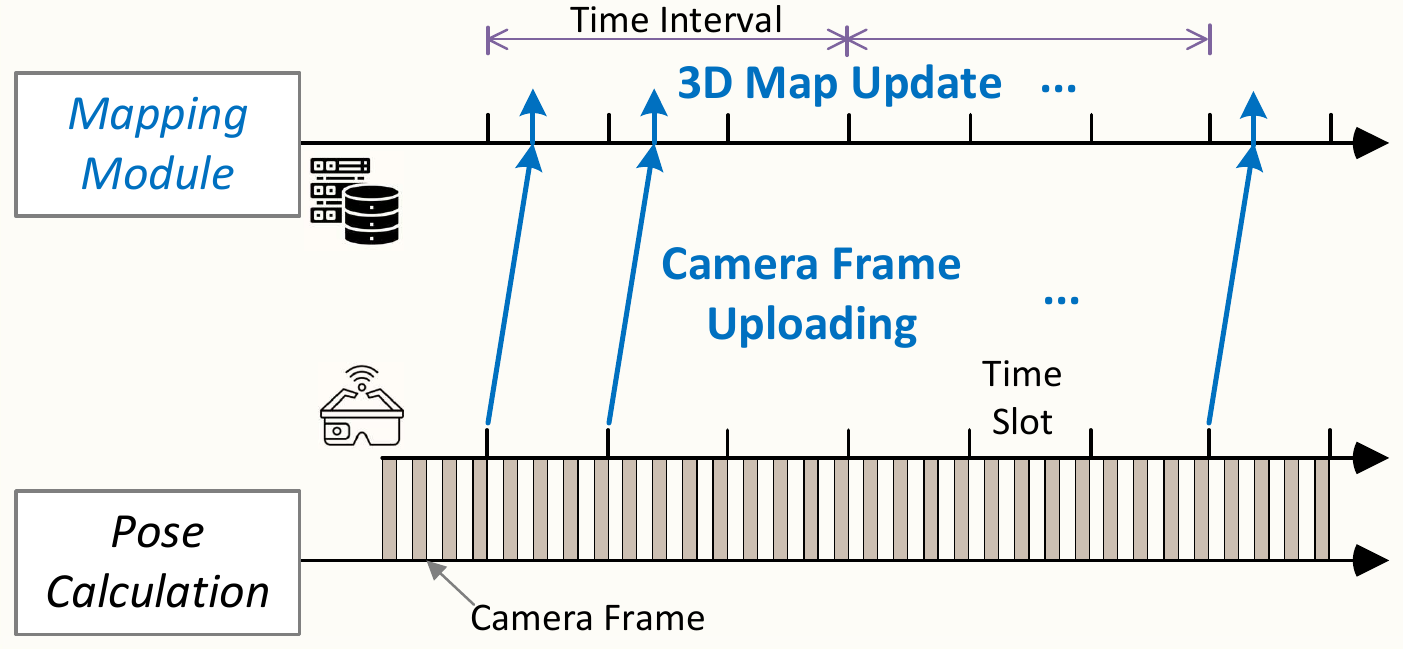}
        \caption{The timeline of 3D map management.}\label{timeline}
    \end{figure}

\subsection{3D Map Model}\label{sec21}

The edge server updates the 3D map per~$F$ camera frames, referred to as a time slot. Denote the set of time slots and the set of camera frames captured across all time slots by~$\mathcal{K}$ and~$\mathcal{F}$, respectively. We illustrate the corresponding timeline of 3D map management in Fig.~\ref{timeline}. Each camera frame, denoted by~$f \in \mathcal{F}$, contains a \bl{set of 3D map} points, denoted by~$\mathcal{M}_{f}$. 

We model the 3D map as a weighted undirected graph to capture the relationships among the camera frames forming the 3D map. The model for a 3D map including four camera frames is illustrated in Fig.~\ref{map}. Denote the 3D map at the beginning of time slot~$k \in \mathcal{K}$ by ~$\mathcal{G}^\text{e}_{k} = (\mathcal{V}^\text{e}_{k}, \mathcal{E}^\text{e}_{k})$, where $\mathcal{V}^\text{e}_{k} \subset \mathcal{F}$ denotes the set of camera frames contained in the 3D map at the beginning of time slot~$k$, and $\mathcal{E}^\text{e}_{k}$ denotes the set of relationships between every pair of camera frames in~$\mathcal{V}^\text{e}_{k}$. For edge~$e = (f, f') \in \mathcal{E}^\text{e}_{k}$ connecting frames~$f \in \mathcal{V}^\text{e}_{k}$ and~$f' \in \mathcal{V}^\text{e}_{k}$, we define its weight as follows:
  \begin{equation}\label{}
        w_{f, f'} = |\mathcal{M}_{f} \cap \mathcal{M}_{f^{'}}|, \,\, \forall f, f' \in \mathcal{V}^\text{e}_{k},
    \end{equation} 
where $|\cdot|$ represents the cardinality of a set, and $\cap$ denotes the intersection of two sets. If the sets of 3D map points contained in camera frames~$f$ and $f^\prime$ are similar, the weight of edge,~$w_{f, f'}$, will be large. As shown in Fig.~\ref{map}, the set of 3D map points corresponding to each camera frame is the collection of the corresponding green points, and the edges are depicted as the orange lines between camera frames.

    \begin{figure}[t]
        \centering
        \includegraphics[width=0.38\textwidth]{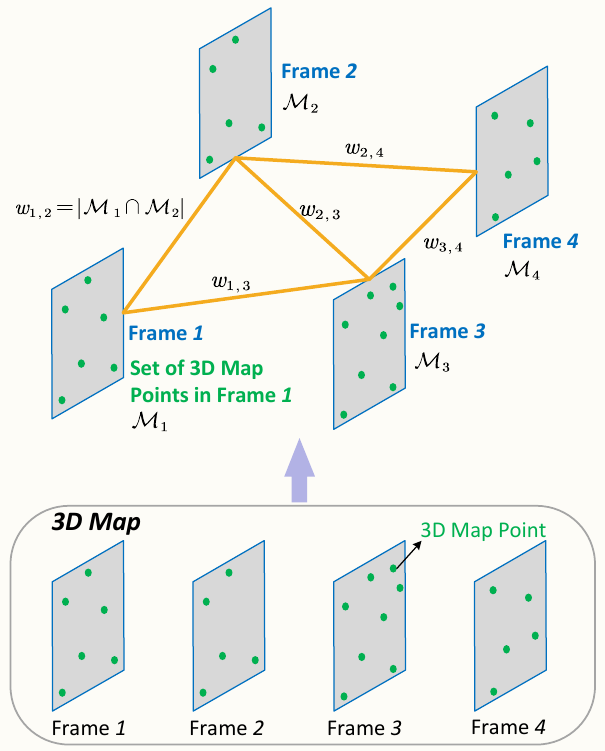}
        \caption{An illustration of the 3D map model.}\label{map}
    \end{figure}

\subsection{Camera Frame Uploading}\label{sec22}

As the MAR user moves around, the MAR device uploads its newly captured camera frames to the edge server for updating the 3D map. The uplink data rate varies over time due to time-varying communication resource availability or communication link quality~\cite{tan2020deep}, and the camera frame uploading must adapt to such variations. As shown in Fig.~\ref{timeline}, we introduce an additional time scale, named the time interval (over a few minutes) so that the dynamics of the uplink data rate is stationary within each time interval. Each time interval consists of~$K$ consecutive time slots, and the set of time slots within time interval~$t$ by~$ \mathcal{K}_{t} = \{ k | (t-1)K < k \leq tK, \forall k \in \mathcal{K} \}$. Within time interval~$t$, we denote the uplink data rate at time slot~$k \in \mathcal{K}_{t}$ by random variable~$d_{k}$, which follows an $N$-state Markov chain. The state transition matrix of the $N$-state Markov chain is assumed to be stationary within each time interval but can vary across time intervals. We introduce random variable~$x_{t}$ to represent the temporal variation of the state transition matrix of the Markov chain across time intervals, i.e., state transition probability $P(d_{k+1}|d_{k}, x_{t})$ is conditioned on $x_t$ and can vary across time intervals.

Without loss of generality, suppose that the MAR device uploads its newly captured camera frames within each time slot to the edge server at the end of the time slot, shown as the blue arrows between the edge server and the MAR device in Fig.~\ref{timeline}. Due to the bandwidth limit for uplink transmissions, a subset of camera frames captured at the end of each time slot needs to be selected for uploading. Denote the set of all camera frames captured during time slot~$k \in \mathcal{K}$ and the subset of camera frames selected for uploading by $\mathcal{F}_{k} \subseteq \mathcal{F}$ and~$\mathcal{U}_{k} \subseteq \mathcal{F}_{k}$, respectively. Let the amount of transmitted data (in bits) for uploading each camera frame, denoted by~$\alpha$, be identical for uploading each camera frame. In any time slot, camera frame uploading should satisfy the following constraint:
    \begin{equation}\label{eq2}
        \alpha  |\mathcal{U}_{k}| \leq d_{k} D^\text{req}, \,\, \forall k \in \mathcal{K},
    \end{equation}
where~$D^\text{req}$ denotes the maximum tolerable transmission delay for uploading the selected camera frames. The value of $D^\text{req}$ in~\eqref{eq2} can be set flexibly according to the overall performance requirement of device pose tracking in MAR.

\subsection{3D Map Update}\label{sec23}

Generally, the mapping module at the edge server involves updating the 3D map and solving a \bl{3D map optimization} problem~\cite{ben2022edge}. In this subsection, we model the 3D map update and the impact of the computing and storage resource limitation for solving a \bl{3D map optimization} problem on 3D map update.

\subsubsection{The impact of resource limitation for solving a \bl{3D map optimization} problem} Given a 3D map~$\mathcal{G}^\text{e}_{k}$, the mapping module is responsible for jointly estimating the 3D device poses corresponding to all camera frames~$f \in \mathcal{V}^\text{e}_{k}$ by solving a \bl{3D map optimization} problem~\cite{campos2021orb}.\footnote{Given a 3D map, solving a \bl{3D map optimization} problem requires finding the maximum likelihood estimations for the 3D device pose corresponding to each camera frame by comparing the feature points contained in every pair of two camera frames in the 3D map~\cite{campos2021orb}.} However, the computing and storage resources required for solving this \bl{3D map optimization} problem increases exponentially with the size of 3D map, while these resources are usually limited at the edge server for any individual MAR device. Considering the resource limitation, we denote the maximum size of the 3D map~by~$V^\text{max}$. Given the set of camera frames in the 3D map, i.e.,~$\mathcal{V}^\text{e}_{k}$, the size of 3D map should satisfy the following constraint:
    \begin{equation}\label{eq4}
        |\mathcal{V}^\text{e}_{k}| \leq V^\text{max}, \,\, \forall k \in \mathcal{K}.
    \end{equation} 
Meanwhile, due to the resource limitation, the 3D map cannot store all the camera frames ever uploaded~\cite{ran2019sharear, campos2021orb}, resulting in the need of removing some camera frames regularly to update the 3D map.

\subsubsection{3D map update} In each time slot, the set of camera frames stored in the 3D map is updated after the edge server receives newly uploaded camera frames.\footnote{The phrases ``update a 3D map'' and ``update the set of frames contained in a 3D map'' are used interchangeable in this paper.} Due to the limited size of the 3D map, a set of camera frames, denoted by $\mathcal{C}_{k} \subseteq \mathcal{V}^\text{e}_{k}$, are removed from the 3D map. Given the set of newly uploaded camera frames~$\mathcal{U}_{k}$, the set of camera frames in the 3D map in time slot~$k+1$, i.e.,~$\mathcal{V}^\text{e}_{k+1}$, evolves as follows: 
    \begin{equation}\label{eq3}
        \mathcal{V}^\text{e}_{k+1} = \left\{ \mathcal{U}_{k} \cup \mathcal{V}^\text{e}_{k} \right\} \backslash \mathcal{C}_{k}, \,\, \forall k, k+1 \in \mathcal{K}.
    \end{equation}
The evolution of set~$\mathcal{V}^\text{e}_{k}$ affects both the nodes and the edges of the graph representing the 3D map. Following the 3D map update, the edge server sends the 3D map back to the MAR device as references for supporting the local 3D pose calculation at the MAR device within the subsequent time slot. 

From Subsections~\ref{sec22} and~\ref{sec23}, it can be seen that 3D map management decisions in each time slot involve both~$\mathcal{U}_{k}$ for camera frame uploading and $\mathcal{C}_{k}$ for 3D map update.

\subsection{Pose Estimation Uncertainty}\label{sec24}

A \bl{performance} metric is required to measure the impact of 3D map management decisions on the performance of device pose tracking~\cite{pan2023joint}. The tracking accuracy is widely used as the metric to evaluate the performance of pose calculation in MAR~\cite{chen2018marvel,ben2022edge}. However, due to the lack of ground truth for 3D poses of an MAR device in real time, a different metric is needed for guiding real-time 3D map management. Recent studies have concentrated on the metric of pose estimation uncertainty, which captures how the quality of a 3D map affects the robustness of MAR or SLAM. For a given 3D map, the pose estimation uncertainty can be obtained in real time. A lower uncertainty represents a higher reliability of the 3D pose estimation, and in turn a higher pose calculation accuracy by reducing cumulative errors~\cite{khosoussi2019reliable}. Existing works have shown the suitability of this metric in facilitating SLAM, especially for tracking accuracy improvement~\cite{chen2023adaptslam}.
 
Since the pose estimation uncertainty affects the reliability of device pose tracking in MAR and can be calculated without requiring ground truth for 3D poses of an MAR device~\cite{rodriguez2018importance}, we adopt it as the performance metric to guide 3D map management. Generally, pose estimation uncertainty, which is unitless, characterizes the impact of the relationship among camera frames in a given 3D map on the error covariance of estimated 3D poses based on this 3D map~\cite{chen2021cramer}. Given 3D map~$\mathcal{G}^\text{e}_{k}$, the pose estimation uncertainty is calculated according to the node connectivity and edge weights in the 3D map, as follows~\cite{chen2023adaptslam,khosoussi2019reliable}:
    \begin{equation}\label{eq6}
        u(\mathcal{G}^\text{e}_{k}) = - \log \left( \det ( \hat{\bm{L}}(\mathcal{G}^\text{e}_{k}) \otimes \boldsymbol{\Pi}) \right), \forall k \in \mathcal{K},
    \end{equation}
where~$\hat{\bm{L}}(\mathcal{G}^\text{e}_{k})$ denotes the reduced Laplacian matrix of the graph representing the 3D map~$\mathcal{G}^\text{e}_{k}$~\cite{godsil2001algebraic},~$\otimes$ represents the Kronecker product, and~$\det(\cdot)$ denotes the determinant of a matrix~\cite{petersen2008matrix}. Matrix~$\boldsymbol{\Pi}$ in~\eqref{eq6} has a dimension of~$6\times6$ due to the six degrees of freedom (DoF) of a 3D pose, and the value of $\boldsymbol{\Pi}$ is related to the camera settings of the MAR device and can usually be assumed as a constant.

\section{Problem Formulation \& Transformation}

In this section, we first formulate a 3D map management problem with the objective of minimizing the pose estimation uncertainty. Then, we transform the problem into a Markov decision process (MDP) problem.

\subsection{Problem Formulation}

To capture the impact of the 3D map updated in time slot~$k$ on the device pose calculation for the camera frames captured in the subsequent time slot~$k+1$, we define the average pose estimation uncertainty over all camera frames in set~$\mathcal{F}_{k+1}$ as~$\upsilon_{k}$. Given the set of camera frames~$\mathcal{F}_{k+1}$ captured within time slot~$k+1$, the value of~$\upsilon_{k}$ is given by:
    \begin{equation}\label{eq8}
       \upsilon_{k} =  |\mathcal{F}_{k+1}|^{-1} \sum_{f \in \mathcal{F}_{k+1}}{u(\mathcal{G}^\text{e}_{k} \cup \{ f \})}, \forall k, k+1 \in \mathcal{K}, 
    \end{equation}
where
    \begin{equation}\label{}
        \mathcal{G}^\text{e}_{k} \cup \{ f \} := \left (\mathcal{V}^\text{e}_{k} \cup \{f\}, \mathcal{E}^\text{e}_{k} \cup \{e=(f,f')| f' \in \mathcal{V}^\text{e}_{k} \} \right ),
    \end{equation}
in which $\{e=(f,f')| f' \in \mathcal{V}^\text{e}_{k} \}$ denotes the set of newly generated edges due to adding camera frame~$f$ to 3D map~$\mathcal{G}^\text{e}_{k}$.

To minimize the pose estimation uncertainty over all time slots, we formulate the following optimization problem:  
    \begin{subequations}\label{p1}
        \begin{align}
            \textrm{P1:} \,\, & \min_{ \{ \mathcal{U}_{k}, \mathcal{C}_{k} \}_{k \in \mathcal{K}} } \sum_{ k \in \mathcal{K}}{ \upsilon_{k} }\\
            \textrm{s.t.} &\,\, \eqref{eq2}, \eqref{eq4}, \eqref{eq3}, \\
            & \,\, \mathcal{U}_{k} \subseteq\mathcal{F}_{k}, \,\, \forall k \in \mathcal{K},\\     
            & \,\, \mathcal{C}_{k} \subseteq \mathcal{V}^\text{e}_{k} \cup \mathcal{U}_{k}, \;\; \forall k \in \mathcal{K}.
        \end{align}
    \end{subequations}
The optimization variables in Problem~P1 are the set of selected camera frames for uploading, i.e.,~$\mathcal{U}_{k}$, and the set of camera frames removed from the original 3D map, i.e.,~$\mathcal{C}_{k}$, in each time slot. Solving Problem P1 is challenging due to three reasons. First, managing the 3D map for any given time slot is an NP-hard fixed-cardinality maximization problem~\cite{chen2023adaptslam}. Regular iterative methods cannot be applied to Problem~P1 due to the unknown~\emph{a priori} information on the set of camera frames captured within subsequent time slots. Second, 3D map management across multiple time slots results in a sequential decision-making problem. Decisions made for the 3D map in one time slot inevitably influence those in subsequent time slots. Making decisions for each time slot independently, without considering their cumulative effect, is not likely to yield the optimal long-term 3D map management. Third, the stochastic uplink data rate is non-stationary across multiple time intervals, subject to the influence of the random variable~$x_{t}$, which exacerbates the challenges.

\subsection{Problem Transformation}\label{sec42}

In this subsection, we first analyze the characteristics of pose estimation uncertainty to reduce the solution space for solving Problem~P1, and then transform the problem into an MDP problem.

\subsubsection{Cardinality of the Optimal Solution Set}

We present the following lemma to show that the pose estimation uncertainty decreases when the number of camera frames forming a 3D map increases.

    \begin{lemma}\label{lemma1}
        Given a connected 3D map~$\mathcal{G}=(\mathcal{V}, \mathcal{E})$, the pose estimation uncertainty~$u(\mathcal{G})$ monotonously decreases with the value of~$|\mathcal{V}|$ when $\det(\boldsymbol{\Pi}) \ge 1$.
    \end{lemma}

    \begin{proof}
        See Appendix~\ref{appendix:lemma1}.    
    \end{proof}

Lemma~\ref{lemma1} allows us to reduce the solution space of Problem~P1. The cardinalities of the optimal~$\mathcal{U}_{k}$ and~$\mathcal{C}_{k}$ for time slot $k, \forall k\in \mathcal{K}$ are derived in Theorem~\ref{theorem1}.

    \begin{theorem}\label{theorem1}
       For time slot~$k$, the cardinalities of the optimal~$\mathcal{U}_{k}$ and~$\mathcal{C}_{k}$, i.e.,~$|\mathcal{U}_{k}|$ and~$|\mathcal{C}_{k}|$, are given by
           \begin{equation}\label{eq10}
                |\mathcal{U}_{k}| = \lfloor \alpha^{-1} D^\text{req} d_{k} \rfloor, \,\, \forall k \in \mathcal{K},         
           \end{equation}
        and
           \begin{equation}\label{eq11}
                |\mathcal{C}_{k}| =  V^\text{max} - |\mathcal{U}_{k}|, \,\, \forall k \in \mathcal{K},         
           \end{equation}
        respectively, where $\lfloor \cdot \rfloor$ represents the floor function.
    \end{theorem}

    \begin{proof}
        Based on Lemma~\ref{lemma1} and \eqref{eq4}, the optimal number of camera frames contained in the 3D map, i.e.,~$|\mathcal{V}^\text{e}_{k}|$, should satisfy:
        \begin{equation}\label{eq12}
            | \mathcal{U}_{k}| + |\mathcal{V}^\text{e}_{k}  | - |\mathcal{C}_{k}| = V^\text{max}, \,\, \forall k \in \mathcal{K}.    
        \end{equation}
        For a given value of~$|\mathcal{C}_{k}|$, the objective function value decreases when the number of uploaded camera frames, i.e.,~$|\mathcal{U}_{k}|$, increases. According to~\eqref{eq2}, the optimal value of~$|\mathcal{U}_{k}|$ is given by~\eqref{eq10}. Given a fixed value~$|\mathcal{U}_{k}|$, the objective function value decreases when fewer camera frames are removed from the 3D map. Therefore, the optimal value of~$|\mathcal{U}_{k}|$ is given by~\eqref{eq11}.
    \end{proof}

According to Theorem~\ref{theorem1}, uploading as many camera frames as possible and removing as few camera frames as possible from the 3D map can decrease the pose estimation uncertainty of 3D map management for MAR. While Theorem~\ref{theorem1} gives the cardinalities of the optimal~$\mathcal{U}_{k}$ and~$\mathcal{C}_{k}$, the optimal sets of camera frames for uploading and for removing are yet to be determined. Consequently, there is a need for an approach to determine the optimal~$\mathcal{U}_{k}$ and~$\mathcal{C}_{k}$, which takes the long-term impacts of time-varying device pose and uplink data rate into account.

\subsubsection{Bayes-adaptive MDP}

Next, we reformulate Problem~P1 as an MDP problem. Denote the state space and the action space of the MDP by $\mathcal{S}$ and $\mathcal{A}$, respectively. Denote~$\bm{\mathcal{G}}^\text{e}_{k} = [ \mathcal{G}^\text{e}_{i} ]_{ k-\tau \leq i \leq k }$, $\bm{\mathcal{F}}_{k} = [\mathcal{F}_{k}]_{ k-\tau \leq i \leq k}$, and $\bm{d}_{k} = [d_{k}]_{ k-\tau \leq i \leq k}$. Let $\bm{s}_{k} = [\bm{\mathcal{G}}^\text{e}_{k}, \bm{\mathcal{F}}_{k}, \bm{d}_{k}] \in \mathcal{S}$ and $\bm{a}_{k} = [\mathcal{U}_{k}, \mathcal{C}_{k}] \in \mathcal{A}$ denote the state and the action, i.e., 3D map management decision, at the beginning of time slot~$t$, respectively. Denote the state transition function by $P(\bm{s}_{k+1}|\bm{s}_{k}, \bm{a}_{k})$. In addition, we define the reward function for time slot~$k$ as the negative of the long-term pose estimation uncertainty defined in~\eqref{eq8}, given by:
    \begin{equation}\label{eq13}
        r_{k} = - \upsilon_{k}, \,\, \forall k \in \mathcal{K}.        
    \end{equation}
With the MDP model and Theorem~\ref{theorem1}, we reformulate Problem~P1 as the following discounted MDP problem for sequential 3D map management decision making in the presence of time-varying device pose and uplink data rate:
    \begin{subequations}\label{p2}
        \begin{align}
            \textrm{P2:} \,\, & \max_{ \{ \bm{a}_{k} \}_{k \in \mathcal{K}} } \sum_{ k \in \mathcal{K} }{ \gamma^{k} r_{k} }\\
            \textrm{s.t.} &\,\, (\ref{p1}\text{b-d}), \eqref{eq10}, \eqref{eq11}
        \end{align}
    \end{subequations}
where~$\gamma \in (0,1)$ is the discount factor for quantifying the long-term impact of an action on the rewards obtained in future time slots~\cite{zhou2022digital}. Our goal is to find a policy, i.e.,~$\pi$, for making proper 3D map management decisions in each state. 

The dynamics in Problem~P2 encompass both variations in device pose and uplink data rate. The device pose is determined solely by human behavior, whereas uplink data rate is mostly determined by network conditions. Therefore, the variations in device pose and uplink data rate can be considered independent. As mentioned in Subsection~\ref{sec22}, due to the unknown random variable~$x_{t}$, the specific parameters of the $N$-state Markov chain may vary across time intervals, thereby resulting in a non-stationary uplink data rate. For tractability, we make the assumption that only the dynamics of uplink data rate is non-stationary in problem~P2, rather than device pose variations. Correspondingly, Problem~P2 becomes a Bayes-adaptive MDP (BAMDP) problem, and the transition probabilities corresponding to the uplink data rate~$P(\bm{d}_{k+1}|\bm{d}_{k}, x_{t}), k \in \mathcal{K}_{t}$, are time-varying, where $x_{t} \sim p(x)$ follows an unknown distribution with some latent parameters~\cite{ghavamzadeh2015bayesian}. To solve the BAMDP problem, we establish a digital twin (DT) for capturing the unknown distribution~$p(x)$, presented in Section~\ref{sec_dt}, and propose a model-based deep reinforcement learning (DRL) method using the DT to adapt to both the time-varying uplink data rate and the device pose.

\section{User Digital Twin}\label{sec_dt}

In this section, we create a DT for an individual MAR device, referred to as a user DT (UDT), to establish a data model that can capture the unknown distribution~$x_{t} \sim p(x)$ in approximating~$P(\bm{d}_{k+1}|\bm{d}_{k}, x_{t})$. Our UDT design evolves from the framework presented in~\cite{shen2021holistic}, with a specific focus on assisting 3D map management in MAR. The UDT, consisting of an \emph{MAR device data profile} and several \emph{UDT functions}, is located at the BS. A network controller is responsible for maintaining and updating the MAR device data profile through the execution of the following four UDT functions: (1) real experience collection, (2) latent feature extraction, (3) artificial experience generation, and (4) UDT update. We illustrate the workflow of the designed UDT in Fig.~\ref{udt}. In the ``User Digital Twin'' segment (to the left of the dashed vertical line in Fig.~\ref{udt}), the real experience of 3D map management, including state, action, reward, and next state, is collected and stored in the MAR device data profile at the end of each time slot. Based on the collected real experiences, a deep variational inference method is used to extract latent features from each real experience and generate artificial experiences based on the extracted latent features, which correspond to the aforementioned UDT functions~(2) and~(3), respectively. The generated artificial experiences are also stored in the MAR device data profile. Meanwhile, the UDT update function (i.e., UDT function~(4)) is used to update the parameters of other UDT functions, such as the weights of the DNNs. Details of the four UDT functions are presented below.

    \begin{figure}[t]
        \centering
        \includegraphics[width=0.5\textwidth]{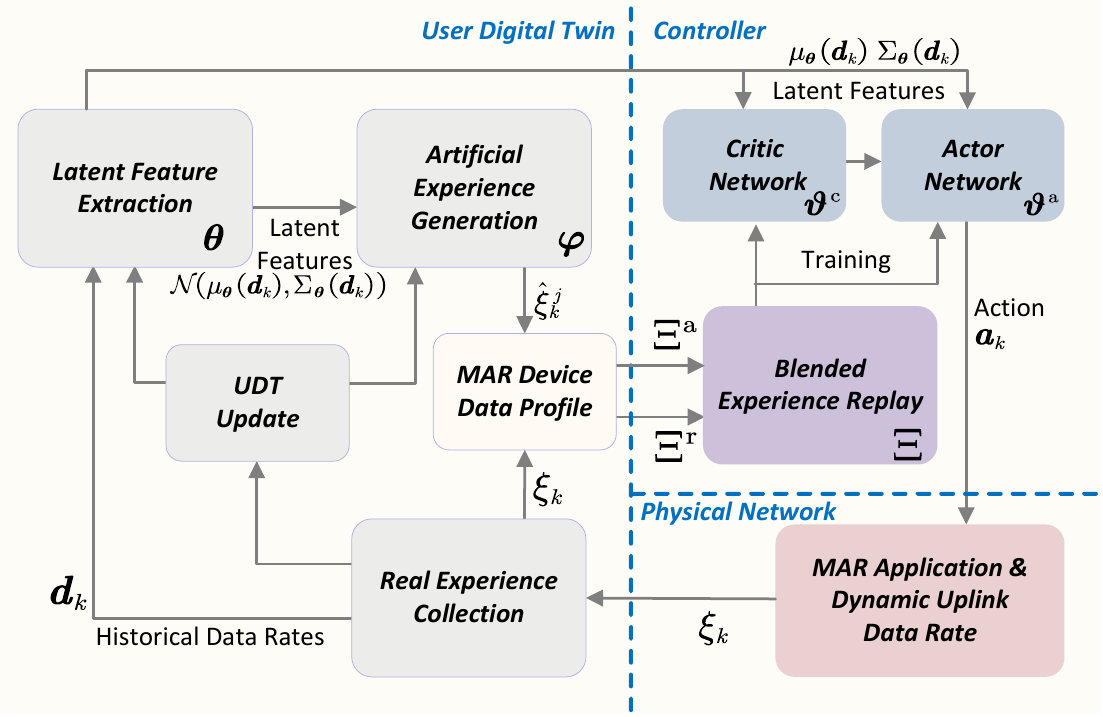}
        \caption{The workflow of the designed UDT and UDT-based 3D map management.}\label{udt}
    \end{figure}

\subsection{Real Experience Collection}

A real experience collected at the beginning of time slot~$k+1$ is the tuple~$\xi_{k} = (\bm{s}_{k}, \bm{a}_{k}, r_{k}, \bm{s}_{k+1})$. In any given time slot, the UDT contains the real experiences of 3D map management collected in preceding time slots for learning the 3D map management policy in subsequent time slots. Let~$\Xi^\text{r}$ denote the set of collected real experiences contained in the UDT, which the UDT can update per time slot by collecting a new real experience. As shown in Fig.~\ref{udt}, the collected data in~$\Xi^\text{r}$ are stored in the MAR device data profile for supporting the other three UDT functions and the decision making on 3D map management.

\subsection{Latent Feature Extraction}\label{sec_dt_2}

Since the distribution~$x_{t} \sim p(x)$ is unknown~\emph{a priori}, calculating~$p(\bm{d}_{k+1} | \bm{d}_{k}, x_{t})$ is not possible. Therefore, we adopt a deep variational inference method to capture the unknown distribution~$p(x)$. 

Define a $Z$-dimensional variable~$\bm{z}$, which follows a normal distribution, i.e.,~$\bm{z} \sim \bm{\mathcal{N}}(\bm{0}, \bm{I}_{Z})$. Given~$\bm{z}$, we introduce a function~$q(\bm{z} | \bm{d}_{k}, \bm{d}_{k+1}; \boldsymbol{\theta})$ parameterized by~$\boldsymbol{\theta}$ to approximate the probability $p(\bm{z} | \bm{d}_{k}, \bm{d}_{k+1}, x_{t})$ and a function~$q( \bm{d}_{k+1} | \bm{d}_{k}, \bm{z}; \boldsymbol{\varphi})$ parameterized by~$\boldsymbol{\varphi}$ to approximate probability $p( \bm{d}_{k+1} | \bm{d}_{k},\bm{z} )$. Since directly maximizing the likelihood~$p(\bm{d}_{k+1}| \bm{d}_{k}; \boldsymbol{\varphi})$ is intractable, our goal is to maximize a lower bound of~$p(\bm{d}_{k+1}|\bm{d}_{k}; \boldsymbol{\varphi})$, given by:
    \begin{equation}\label{eq18}
        \begin{aligned} 
             \log p(\bm{d}_{k+1} | \bm{d}_{k}; \boldsymbol{\varphi}) \ge & \mathbb{E}_{\bm{z} \sim q(\bm{z} | \bm{d}_{k}, \bm{d}_{k+1}; \boldsymbol{\theta})} \left[\log q(\bm{d}_{k+1} | \bm{d}_{k}, \bm{z}; \boldsymbol{\varphi}) \right] \\
            & \,\,\,\,  - l_\text{KL}(q(\bm{z} | \bm{d}_{k}, \bm{d}_{k+1}; \boldsymbol{\theta}) || p(\bm{z})),  
        \end{aligned} 
    \end{equation}
where $l_\text{KL}(q(\bm{z} | \bm{d}_{k}, \bm{d}_{k+1}; \boldsymbol{\theta}) || p(\bm{z}))$ denotes the Kullback–Leibler (KL) divergence:
    \begin{equation}\label{eq19}
        \begin{aligned} 
        & l_\text{KL}(q(\bm{z} | \bm{d}_{k}, \bm{d}_{k+1}; \boldsymbol{\theta}) || p(\bm{z})) = \\ 
        & \,\,\,\, \mathbb{E}_{\bm{z} \sim q(\bm{z} | \bm{d}_{k}, \bm{d}_{k+1}; \boldsymbol{\theta})} \left[ \log q(\bm{z} | \bm{d}_{k}, \bm{d}_{k+1}; \boldsymbol{\theta}) - \log  p(\bm{z}) \right].
        \end{aligned} 
    \end{equation}
We omit the full derivation of the formula in~\eqref{eq18} and refer interested readers~\cite{sachdeva2019sequential}.

Considering that a multi-dimensional normal distribution can be used to capture various data distributions, we adopt a normal distribution~$\bm{\mathcal{N}}(\bm{z} | \mu_{\boldsymbol{\theta}}(\bm{d}_{k}), \Sigma_{\boldsymbol{\theta}}(\bm{d}_{k}))$ to approximate~$q(\bm{z} | \bm{d}_{k}, \bm{d}_{k+1}; \boldsymbol{\theta})$ so that the term~$l_\text{KL}(q(\bm{z} | \bm{d}_{k}, \bm{d}_{k+1}; \boldsymbol{\theta}) || p(\bm{z}))$ in \eqref{eq19} has a closed form~\cite{kingma2013auto}. Specifically, we leverage a DNN to output the parameters of the $Z$-dimensional normal distribution~$(\bm{z} | \bm{d}_{k}, \bm{d}_{k+1}; \boldsymbol{\theta})$, i.e., mean $\mu_{\boldsymbol{\theta}}(\bm{d}_{k})$ and variance $\Sigma_{\boldsymbol{\theta}}(\bm{d}_{k})$. Given parameter~$\boldsymbol{\theta}$, we refer to the output $\mu_{\boldsymbol{\theta}}(\bm{d}_{k})$ and $\Sigma_{\boldsymbol{\theta}}(\bm{d}_{k})$ as the \emph{latent features} representing the unknown distribution~$x_{t} \sim p(x)$, extracted from input~$\bm{d}_{k}$. In addition, we approximate the function~$q( \bm{d}_{k+1} | \bm{d}_{k}, \bm{z}; \boldsymbol{\varphi})$ through another DNN parameterized by~$\boldsymbol{\varphi}$. Given~$\bm{d}_{k}$, the only element in $\bm{d}_{k+1}$ that needs to be predicted is $d_{k+1}$. The value of $d_{k+1}$, denoted by~$\hat{d}_{k+1}$, can be output as follows: 
    \begin{equation}\label{eq20}
        \hat{d}_{k+1} = \phi_{\boldsymbol{\varphi}} (\bm{\tilde{z}}), \,\, \forall k \in \mathcal{K}_{t}.
    \end{equation}
where~$\phi_{\boldsymbol{\varphi}} (\cdot)$ denotes the DNN parameterized by~$\boldsymbol{\varphi}$ with a softmax activation function, and vector~$\bm{\tilde{z}}$ is a sample from the distribution $\bm{\mathcal{N}}(\bm{z} | \mu_{\boldsymbol{\theta}}(\bm{d}_{k}), \Sigma_{\boldsymbol{\theta}}(\bm{d}_{k}))$. Consequently, the outputs of the DNNs with parameters~$\boldsymbol{\vartheta}$ and $\boldsymbol{\varphi}$ are the extracted latent features, i.e.,~$\mu_{\boldsymbol{\theta}}(\bm{d}_{k})$ and $\Sigma_{\boldsymbol{\theta}}(\bm{d}_{k})$, and the value of~$\hat{d}_{k+1}$, respectively. Both outputs are used to facilitate our proposed DRL methods for solving Problem~P2, as shown in Fig.~\ref{udt}.

\subsection{Artificial Experience Generation}

The artificial experience generation function of the UDT can generate data samples for~$\bm{d}_{k+1}$ by leveraging~$\bm{d}_{k}$ and the extracted latent features~$\mu_{\boldsymbol{\theta}}(\bm{d}_{k})$ and $\Sigma_{\boldsymbol{\theta}}(\bm{d}_{k})$ at time slot~$k$. Upon feeding~$\bm{d}_{k}$ into the two DNNs parameterized by~$\boldsymbol{\vartheta}$ and~$\boldsymbol{\varphi}$ sequentially, the output of the Softmax activation function in the DNN represented by~\eqref{eq20} yields a probability vector. This vector encapsulates the likelihood of each potential state among the $N$ states in the Markov chain pertaining to the uplink data rate, as mentioned in Subsection~\ref{sec22}. As a result, the artificial experience generation function of the UDT can generate a set of data samples, each of which represents a possible uplink data rate within time slot~$k+1$, denoted by~$\hat{d}_{k+1}^{j}$. Let~$\hat{\bm{d}}_{k+1}^{j} = [\bm{d}_{k}, \hat{d}_{k+1}^{j}]$ denote a generated data sample according to the collected real data~$\bm{d}_{k}$. 

Based on each generated sample of the uplink data rate, this UDT function can generate an artificial experience, denoted by~$\hat{\xi}_{k}^{j} = (\hat{\bm{s}}_{k}^{j}, \hat{\bm{a}}_{k}^{j}, \hat{r}_{k}^{j}, \hat{\bm{s}}_{k}^{j+1})$. Specifically, given state~$\hat{\bm{s}}_{k}^{j}$, this UDT function can randomly select an action, denoted by $\hat{\bm{a}}_{k}^{j}$, which satisfies the constraints in Problem~P2 and is used for the emulation of 3D map management. Subsequently, this UDT function can emulate the 3D map at time slot~$k+1$, denoted by~$\hat{\bm{\mathcal{G}}}^{\text{e},j}_{k+1}$, as~\eqref{eq3} given action~$\hat{\bm{a}}_{k}^{j}$ taken in state~$\hat{\bm{s}}_{k}^{j}$ and calculate the reward using~\eqref{eq13}. As a result, given any tuple of real experience~$\xi_{k} = (\bm{s}_{k}, \bm{a}_{k}, r_{k}, \bm{s}_{k+1}) \in \Xi^\text{r}$, this UDT function can generate $J$ artificial experiences with the~$j$\,th tuple given $(\bm{s}_{k}, \hat{\bm{a}}_{k}^{j}, \hat{r}_{k}^{j}, \hat{\bm{s}}_{k}^{j+1} )$ where:
    \begin{equation}\label{eq21}
        \begin{aligned}
            \hat{\bm{s}}_{k}^{j+1} &  = [\hat{\bm{\mathcal{G}}}^{\text{e},j}_{k+1}, \bm{\mathcal{F}}_{k+1}, \hat{\bm{d}}_{k+1}^{j}]. 
        \end{aligned}
    \end{equation}
We define the set of these generated artificial experiences as~$\Xi^\text{a} = \left\{\hat{\xi}_{k}^{j} | \forall 0< j \le J, \xi_{k} \in \Xi^\text{r} \right\}$. As shown in Fig.~\ref{udt}, the data of generated artificial experiences are also stored in the MAR device data profile and can be used to support making 3D map management decisions.

\subsection{UDT Update}

    \begin{algorithm}[t] 
        \caption{UDT Update}\label{alg1}
        \LinesNumbered
        \textbf{Input:} $\Xi^\text{r}$\\
        \textbf{Initialization:} $\boldsymbol{\theta}$, $\boldsymbol{\varphi}$\\
         $\boldsymbol{\theta}_{*}$, $\boldsymbol{\varphi}_{*}$ $\leftarrow$ Update $\boldsymbol{\theta}$ and $\boldsymbol{\varphi}$ by optimizing~\eqref{eq22} based on real experiences $\Xi^\text{r}$;\\
        $\Xi^\text{a}$ $\leftarrow$ Update with the newly generated tuples~$\left\{\hat{\xi}_{k}^{j} | \forall 0< j \le J, \xi_{k} \in \Xi^\text{r} \right\}$ using~\eqref{eq21};\\
        \textbf{Output:} $\boldsymbol{\theta}_{*}$, $\boldsymbol{\varphi}_{*}$, $\Xi^\text{a}$
    \end{algorithm} 

Both the latent feature extraction and the artificial experience generation functions require online optimization of parameters~$\boldsymbol{\theta}$ and~$\boldsymbol{\varphi}$. The optimization is conducted by minimizing the loss function given by the right-hand side of~\eqref{eq18} as follows: 
     \begin{equation}\label{eq22}
        \begin{aligned}
            L(\boldsymbol{\theta}, \boldsymbol{\varphi}) =  & \mathbb{E}_{\bm{z} \sim q(\bm{z} | \bm{d}_{k}, \bm{d}_{k+1}; \boldsymbol{\theta})} \left[\log q(\bm{d}_{k+1} | \bm{d}_{k}, \bm{z}; \boldsymbol{\varphi}) \right] \\
            &  - l_\text{KL}(q(\bm{z} |\bm{d}_{k+1}, \bm{d}_{k}; \boldsymbol{\theta}) || p(\bm{z})).
        \end{aligned}
    \end{equation}
The gradient descent method can be employed to find the optimal parameters~$\boldsymbol{\theta}_{*}$ and~$\boldsymbol{\varphi}_{*}$, utilizing the reparametrization trick in~\cite{zintgraf2019varibad}, that minimize the gradients of~$L(\boldsymbol{\theta}, \boldsymbol{\varphi})$ for DNN backward propagation. We introduce the update of DNNs with parameters~$\boldsymbol{\theta}$ and~$\boldsymbol{\varphi}$ and the generation of artificial experiences in Algorithm~\ref{alg1}. In Line~3, we optimize the parameters~$\boldsymbol{\theta}$ and~$\boldsymbol{\varphi}$ by minimizing the loss function in~\eqref{eq22}. In Line~4, given a newly collected set of real experiences, i.e.,~$\Xi^\text{r}$, the UDT generates a new set of artificial experiences for adapting to the dynamic network uplink data rate. Our UDT design allows for flexible adjustments of parameters~$\boldsymbol{\theta}$ and $\boldsymbol{\varphi}$ according to the newly collected real experiences, which happens once every~$W$ time slots.

\section{Model-based DRL Method}\label{sec_rl}

While the designed UDT can capture the dynamics of the uplink data rate, we still need a method to adapt to temporal variations in the \bl{device} pose and the uplink data rate when solving Problem~P2. Conventional DRL-based methods can be used to solve MDP problems characterized by unknown but stationary environmental dynamics~\cite{ghavamzadeh2015bayesian}. However, they cannot be directly applied to solving Problem~P2 due to the non-stationary nature of the MDP across time intervals.

Therefore, we propose a \emph{model-based} DRL (MBRL) method by leveraging extensive organized data provided by the UDT, including the extracted latent features of the uplink data rate and the generated artificial experiences, to solve the BAMDP problem. Our proposed method is built upon an off-policy model-free DRL algorithm, in which historical experiences can be used offline for training the DNNs~\cite{zhou2022digital,cheng2019space}. Next, we will present the model-free DRL framework and our designs of UDT-assisted DRL after that.

\subsection{Model-free DRL}

Given that the UDT can extract latent features regarding dynamic uplink data rate through the approximation of random variable~$x_{t}$, we extend the originally state~$\bm{s}_{k}$ by defining the augmented state as~$\dot{\bm{s}}_{k} = [\bm{s}_{k}, \mu_{\boldsymbol{\theta}}(\bm{d}_{k}), \Sigma_{\boldsymbol{\theta}}(\bm{d}_{k})]$. The latent features extracted by using the UDT can represent the information on the random variable~$x_{t}$, which results in the non-stationary uplink data rate. Therefore, adopting the augmented state can transform the BAMDP into an MDP, thereby enabling us to apply a DRL-based method for learning a 3D map management policy. 

Using the augmented state, we adopt an actor-critic framework to learn the optimal policy~$\pi^{*}$. Define a Q-value function of state~$\dot{\bm{s}}_{k}$ and action~$\bm{a}_{k}$ as the accumulated discounted reward, as follows: 
    \begin{equation}\label{}
        Q(\dot{\bm{s}}_{k}, \bm{a}_{k}) =  \sum_{j = 1}^{J}{\gamma^{j} r_{k+j+1} }, \,\, \forall k \in \mathcal{K},
    \end{equation}
where the Q-value quantifies the long-term impact of each action on subsequent states and actions~\cite{zhou2022digital}. However, calculating the Q-value directly is impossible due to unknown state transitions in the future. Thus, we derive the Q-value calculated for time slot~$k$ by using the obtained reward at time slot~$k$ and the Q-value calculated for time slot~$k+1$, as follows:
    \begin{equation}\label{eq24}
        Q(\dot{\bm{s}}_{k}, \bm{a}_{k}) = r_{k} + \gamma Q(\dot{\bm{s}}_{k+1}, \bm{a}_{k+1}), \forall k \in \mathcal{K}.
    \end{equation}
Given~\eqref{eq24}, model-free DRL methods approximate the Q-value function by minimizing the temporal difference between the Q-values for different time slots~\cite{cheng2019space,luo2022deep}. Specifically, a DNN (the critic network) with the parameter~$\boldsymbol{\vartheta}^\text{c}$ is leveraged to approximate the Q-value function, i.e.,~$Q(\bm{s}, \bm{a} ; \boldsymbol{\vartheta}^\text{c})$. The loss function for optimizing the parameter~$\boldsymbol{\vartheta}^\text{c}$ is given by:
    \begin{equation}\label{eq25}
        \begin{aligned}
            & L(\boldsymbol{\vartheta}^\text{c}) = \\
            & \,\,\,\, \frac{1}{|\Xi|} \sum_{\xi_{k} \in \Xi} \left( r_{k} + \gamma Q(\dot{\bm{s}}_{k+1}, \pi(\dot{\bm{s}}_{k+1}); \boldsymbol{\vartheta}^\text{c}) - Q(\dot{\bm{s}}_{k}, \bm{a}_{k}; \boldsymbol{\vartheta}^\text{c}) \right)^{2}, 
        \end{aligned}
    \end{equation}
where $\Xi$ denotes a batch of tuples $\xi_{k}$ selected from historical experiences for training, and $|\Xi|$ represents the batch size. Note that we consider an off-policy DRL framework, wherein the policy $\pi(\dot{\bm{s}}_{k})$ employed for Q-value approximation in~\eqref{eq25} may be different from the policy used for actual action execution in practice.

    \begin{algorithm}[t] 
        \caption{AMM Algorithm}\label{alg2}
        \LinesNumbered
        \textbf{Input:} $W$, $I$, $|\Xi|$\\
        \textbf{Initialization:} $\boldsymbol{\vartheta}^\text{c}$, $\boldsymbol{\vartheta}^\text{a}$, $\boldsymbol{\theta}$, $\boldsymbol{\varphi}$, $\Xi^\text{r}$, $\Xi^\text{a}$, ${\bm{s}}_{1}$\\
        \For{$k \in \mathcal{K}$}
        {   
            $\mu_{\boldsymbol{\theta}}(\bm{d}_{k})$, $\Sigma_{\boldsymbol{\theta}}(\bm{d}_{k})$ $\leftarrow$ the UDT extracts the latent features as~Subsection\ref{sec_dt_2};\\
            Select action $\bm{a}_{k} = \pi \left( \dot{\bm{s}}_{k}; \boldsymbol{\vartheta}^\text{a} \right)$;\\
            $r_{k}$, $\bm{s}_{k+1}$ $\leftarrow$ take action $\bm{a}_{k} $ on state~$\bm{s}_{k}$;\\
            
            Combine $I$ tuples randomly selected from $\Xi^\text{r}$ and $|\Xi|-I$ tuples randomly selected from $\Xi^\text{a}$ as $\Xi$;\\
            Update $\boldsymbol{\vartheta}^\text{c}$ by minimizing~\eqref{eq25};\\
            Update $\boldsymbol{\vartheta}^\text{a}$ by using policy gradient descent in~\eqref{eq26};\\

            $\Xi^\text{r}$ $\leftarrow$ Update with the latest tuples of real experiences, i.e.,~$\{ \xi_{j} | k-|\Xi^\text{r}| < j \leq k \}$;\\ 

            \If{$k \mod W = W-1$}
            {   
                $\boldsymbol{\theta}_{*}$, $\boldsymbol{\varphi}_{*}$, $\Xi^\text{a}$ $\leftarrow$ Run Algorithm~\ref{alg1} for UDT update;\\
            } 
            
            $k$, $\bm{s}_{k}$ $\leftarrow$ $k+1$, $\bm{s}_{k+1}$;\\       
        }
        \textbf{Output:} $\pi(\dot{\bm{s}}_{k};\boldsymbol{\vartheta}^\text{a}_{*})$

    \end{algorithm} 

Given an approximated Q-value function, our goal is to find a policy~$\pi(\dot{\bm{s}}_{k})$ for taking an action in each state with the consideration of the long-term impact of the action. Model-free DRL can be used to find the parameters of the optimal policy through parameterizing the policy~$\pi(\dot{\bm{s}}_{k})$. Specifically, we approximate the 3D map management policy~$\pi(\dot{\bm{s}}_{k}; \boldsymbol{\vartheta}^\text{a})$ by using another DNN (the actor network) with parameter~$\boldsymbol{\vartheta}^\text{a}$. Given~$Q(\dot{\bm{s}}_{k}, \bm{a}_{k}; \boldsymbol{\vartheta}^\text{c})$, the policy gradient calculated for optimizing the policy~$\pi(\dot{\bm{s}}_{k};\boldsymbol{\vartheta}^\text{a}_{*})$ is given as follows:
    \begin{equation}\label{eq26}
        \begin{aligned}
            & \nabla_{\boldsymbol{\vartheta}^\text{a}} \Omega (\boldsymbol{\vartheta}^\text{a}) = \\
            & \,\,\,\, \frac{1}{|\Xi|} \sum_{\xi_{k} \in \Xi} \nabla_{\mathbf{w}} Q \left(\dot{\bm{s}}_{k}, \bm{a}_{k} ; \boldsymbol{\vartheta}^\text{c} \right) \big|_{\pi(\dot{\bm{s}}_{k}; \boldsymbol{\vartheta}^\text{a})} \nabla_{ \boldsymbol{\vartheta}^\text{a} }\pi\left( \dot{\bm{s}}_{k}; \boldsymbol{\vartheta}^\text{a} \right),
            \end{aligned}
    \end{equation}
where~$\Omega (\boldsymbol{\vartheta}^\text{a})$ represents the value of the objective function (\ref{p2}a) achieved with the 3D map management policy~~$\pi(\dot{\bm{s}}_{k}; \boldsymbol{\vartheta}^\text{a})$.

\subsection{Blended Experience Replay}

With state augmentation, the expansive state space poses a challenge for a model-free DRL-method in finding the optimal policy. This is because the required volume of collected real experiences for policy learning significantly increases with the state space size. Therefore, our idea behind the proposed MBRL method is using both the collected real experiences and the generated artificial experiences provided by the UDT for accelerating the policy learning. Specifically, the proposed MBRL method incorporates a new mechanism, called \emph{blended experience replay}, to accelerate the training of the actor and critic networks by using both types of experiences from the UDT. Correspondingly, the batch used for DNN training in~\eqref{eq25} and~\eqref{eq26} can be selected from~$\Xi^\text{a}$, $\Xi^\text{r}$, or both.

We propose an adaptive map management (AMM) algorithm using the UDT data in Algorithm~\ref{alg1}. In Line~4, we use the UDT to extract the latent features~$\mu_{\boldsymbol{\theta}}(\bm{d}_{k})$ and $\Sigma_{\boldsymbol{\theta}}(\bm{d}_{k})$ for enabling state augmentation as mentioned in Subsection~\ref{sec_dt_2}. This procedure corresponds to the grey arrows crossing from the ``User Digital Twin'' segment to the ``Controller'' segment in Fig.~\ref{udt}. In Lines~5-6, the decision on the output action~$\bm{a}_{k}$ is made by the actor network given state~$\dot{\bm{s}}_{k}$, and the corresponding reward and the next state are observed. Next, we train the actor and critic networks by using the blended experience replay, as shown in Lines~7-9. Specifically, we randomly select~$I$ tuples of collected real experiences from $\Xi^\text{r}$ and~$|\Xi|-I$ tuples of generated artificial experiences from~$\Xi^\text{a}$, respectively, and combine them in a batch to update the parameters of the actor and the critic networks. The procedure of using the UDT to support blended experience replay corresponds to the grey arrows from the ``MAR Device Data Profile'' block to the ``Blended Experience Replay'' block in Fig.~\ref{udt}. The parameter~$I$ controls the ratio of the amount of artificial experiences to the overall batch size in training. In Line~11, the UDT updates the set of real experiences~$\Xi^\text{r}$ by adding the newly collected \bl{tuples} into the set. This step corresponds to the grey arrow crossing from the ``Physical Network'' segment to the ``User Digital Twin'' segment in Fig.~\ref{udt}.

In addition to training the DNNs for MBRL, we adopt a resource-efficient online training method to update the UDT as given in lines~11-13. We introduce a hyper-parameter~$W$ for UDT update, and optimize the parameters~$\boldsymbol{\theta}$ and $\boldsymbol{\varphi}$ per~$W$ time slots as mentioned in~Section~\ref{sec_dt}. Meanwhile, the UDT generates new artificial experiences according to the newly collected real experiences, which increases the adaptivity of the proposed MBRL method.

\section{Numerical Results}

\begin{table}[t]
    \footnotesize 
    \centering
    \captionsetup{justification=centering,singlelinecheck=false}
    \caption{Simulation Parameters}\label{table1}
    \begin{tabular}{c|c|c|c}
        \hline\hline
         Parameter & Value & Parameter & Value\\
         \hline\hline
         $D^\text{req}$ & 0.5\,second & $\alpha$ & 5\,Mbits \\
         \hline
         $d_{k}^\text{max}$& 80\,Mbits/second & $d_{k}^\text{min}$ & 40\,Mbits/second\\
         \hline
         $V^\text{max}$ & 25-45\,frames& $F$ & 60\,frames\\
         \hline
    \end{tabular}
\end{table}

In our simulations, we use the ``westgate-playroom'' camera frame sequence in the SUN3D dataset~\cite{xiao2013sun3d}, which contains data collected in a real indoor environment. The set of 3D map points detected in each camera frame~$\mathcal{M}_{f}$ is obtained using the open-source ORB-SLAM framework~\cite{campos2021orb}. Important parameter settings are listed in Table~I.

We utilize two long-short-term-memory (LSTM) layers with 300 neurons, followed by four fully connected layers with (256, 128, 128, 32) neurons, to build the DNN with parameter~$\boldsymbol{\theta}$ for latent feature extraction in the UDT. Meanwhile, we use four fully connected layers with (256, 128, 64, 16) neurons to build the DNN with parameter~$\boldsymbol{\varphi}$ for artificial experience generation in the UDT. 

For the actor and the critic DNNs used in the MBRL scheme, we leverage two graph convolutional networks (GCNs) as embedding layers to capture the relationship among camera frames in the 3D map. Each GCN consists of two graph convolutional layers with (128, 32) neurons. Following the embedding layer, three fully connected layers with (64, 32, 32) neurons and four fully connected layers with (64, 32, 16, 4) neurons are used for building the critic and the actor DNNs, respectively. 

We adopt the following 3D map management schemes in MAR as benchmark~\cite{apicharttrisorn2020characterization, campos2021orb,chen2023adaptslam}:
    \begin{itemize}
        \item \emph{Latest Frame First (LFF)}: The 3D map is periodically updated by adding the lastly captured camera frames and removing the camera frames captured the earliest;

        \item \emph{Periodical Uploading (PU)}: The camera frames to upload are selected uniformly from all camera frames captured within each time slot, and the earliest captured camera frames are removed from the 3D map; 

        \item \emph{ADAPT}~\cite{chen2023adaptslam}: The set of camera frames contained in the 3D map is selected from all camera frames captured within each time slot, by using an optimization method proposed to minimize the uncertainty.
    \end{itemize}

\subsection{Performance of UDT}

In this subsection, we compare the performance of the proposed UDT-based approach with that of LSTM-based and Markov model-driven approaches in capturing the dynamics of the uplink data rate. 

    \begin{figure}[t]
        \centering
        \begin{subfigure}[b]{0.45\textwidth}
            \includegraphics[width=\textwidth]{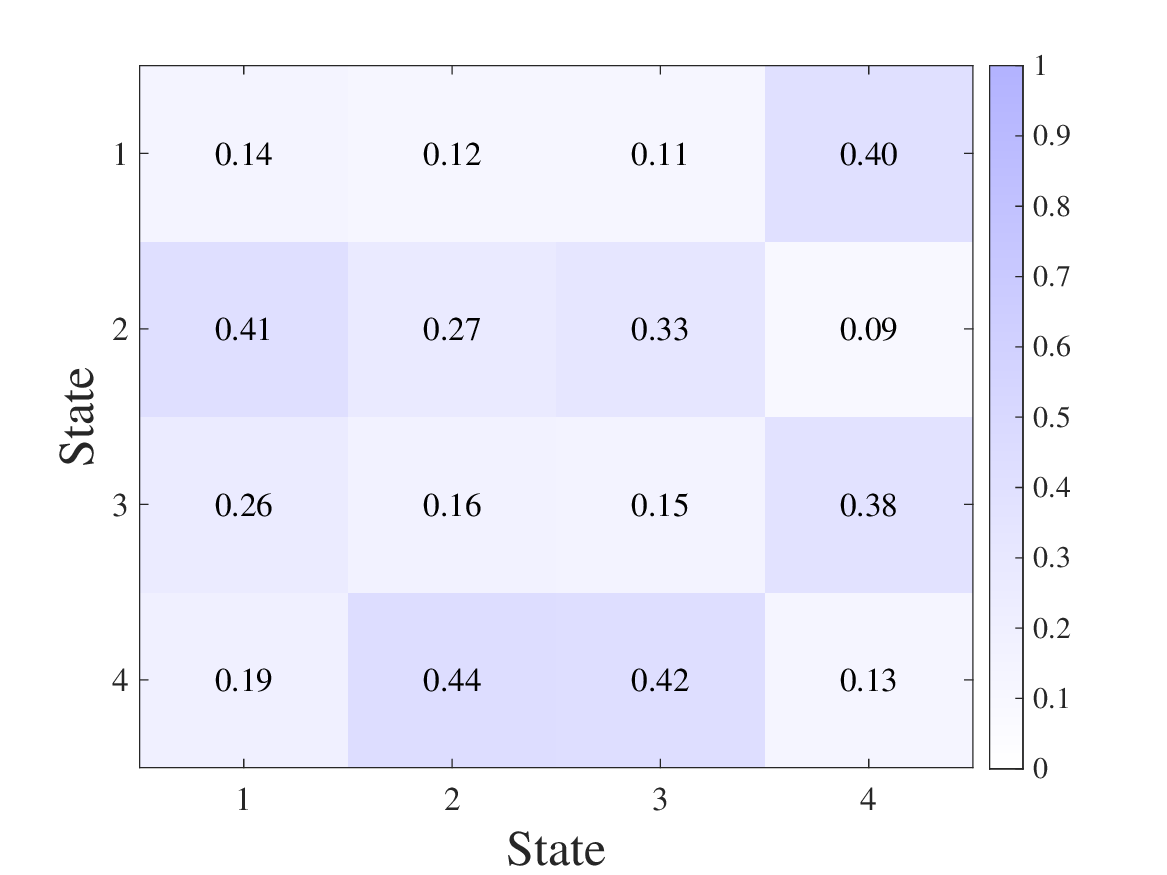}
            \caption{The estimated values using the UDT.}
            \label{fig:sub1}
        \end{subfigure}
        \quad
        \begin{subfigure}[b]{0.45\textwidth}
            \includegraphics[width=\textwidth]{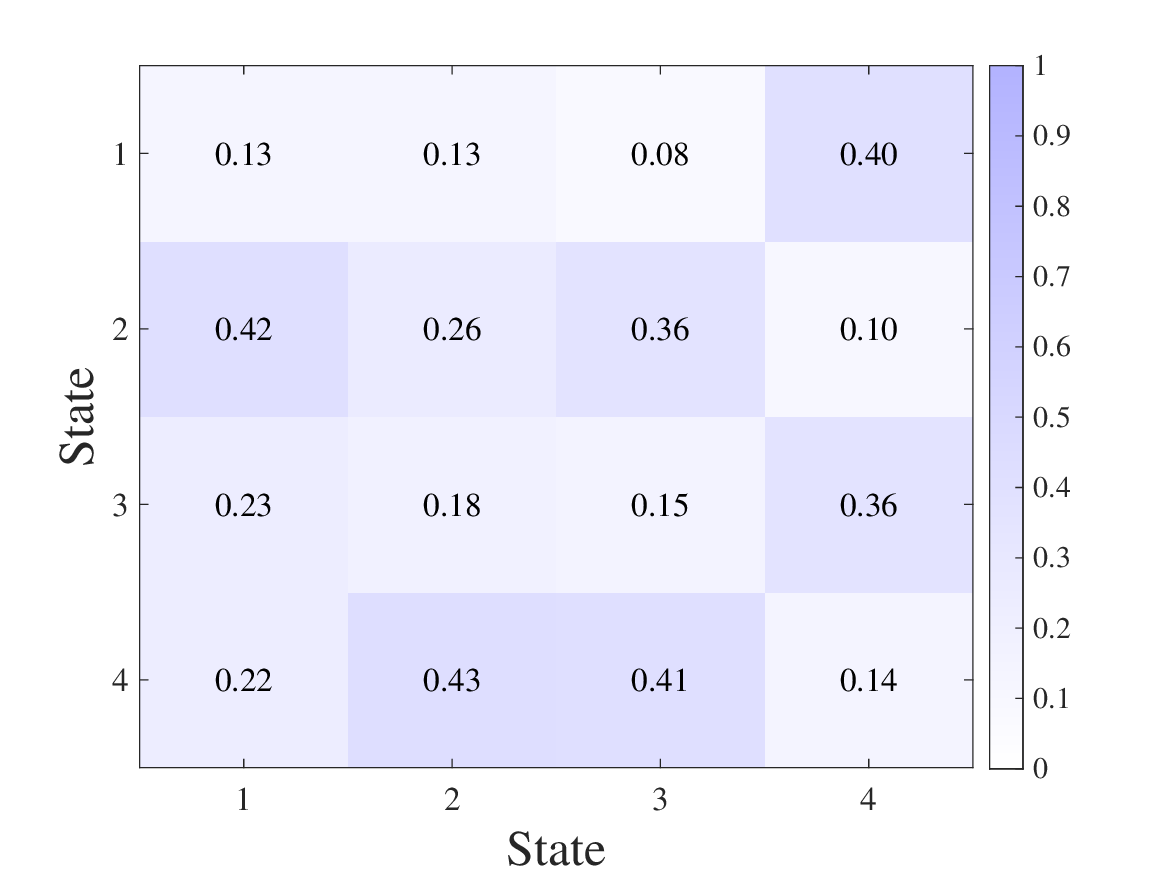}
            \caption{The actual values.}
            \label{fig:sub2}
        \end{subfigure}
        \caption{The estimated and the actual values of the state transition matrix of a $4$-state Markov chain within one time interval.}
        \label{fig34}
    \end{figure} 

We first show the accuracy of the designed UDT in capturing the dynamics of the time-varying uplink data rate within one time interval. In Fig.~\ref{fig34}, a comparison is made between the estimated and actual values of the state transition matrix for a 4-state Markov chain within one time interval. Each state in the Markov chain corresponds to a distinct uplink data rate,~i.e.,~$d_{k}$, and each value in Fig.~\ref{fig34} shows a state transition probability. We can observe that the values estimated by the UDT are very close to the actual values of the state transition matrix, which underlines the effectiveness of the UDT in capturing the stationary uplink data rate within one time interval.

    \begin{figure}[t]
        \centering
        \includegraphics[width=0.45\textwidth]{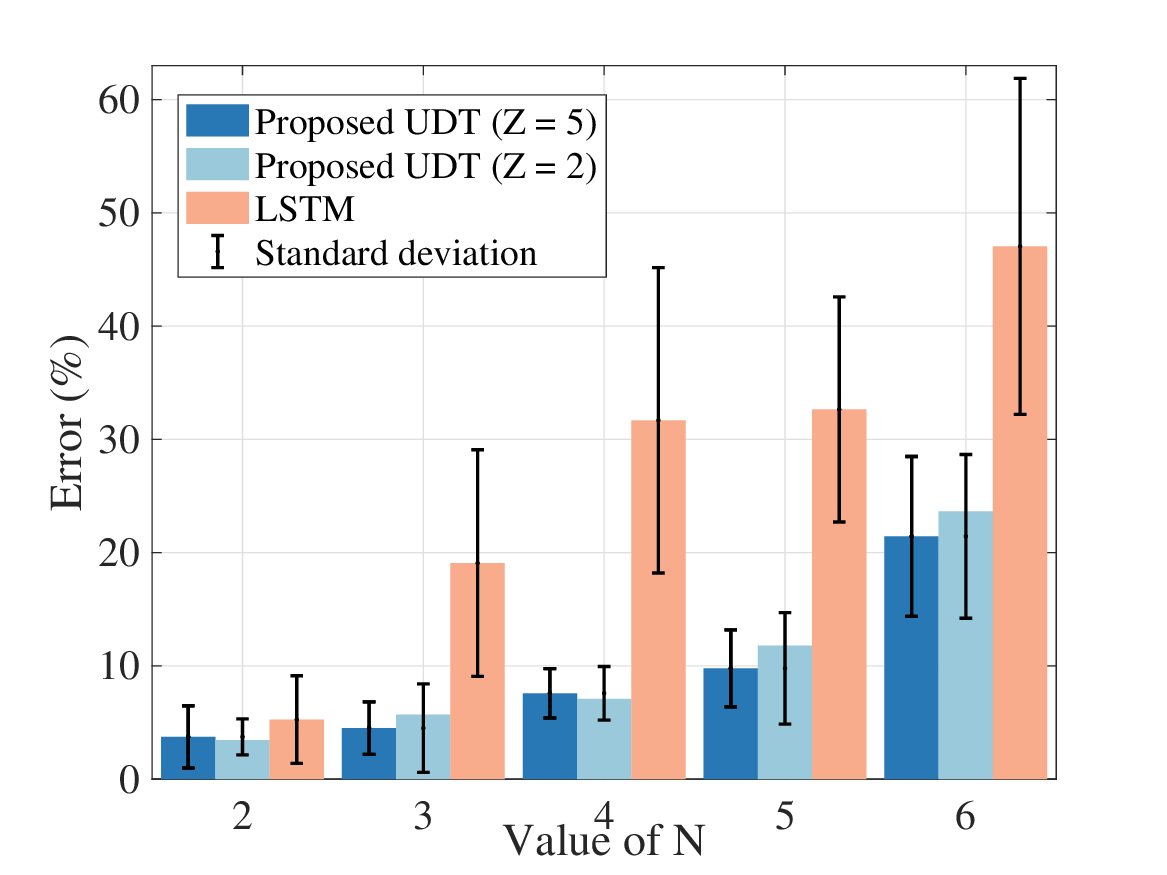}
        \caption{Performance comparison between the UDT and the LSTM-based prediction across three time intervals.}\label{fig2}
    \end{figure}

Next, we compare the performance of two designed UDTs with latent features of different dimensions, i.e.~$Z$, with that of a data-driven method (labeled as ``LSTM'') in capturing the non-stationary dynamics of the uplink data rate across multiple time intervals. Specifically, three time intervals are considered, each of which corresponds to a unique state transition matrix. As shown in Fig.~\ref{fig2}, we plot the error of the estimated state transition matrix for the three schemes versus the number of states in the Markov chain, i.e.,~$N$. The value of each bar represents the average error over 20 independent simulation runs. We can observe that the designed UDT outperforms the LSTM-based method in all scenarios since the LSTM-based method simply makes deterministic predictions rather than capturing the underlying state transition probabilities. Additionally, given the UDT, the error increases with~$N$ since a larger value of~$N$ results in more parameters to approximate.

    \begin{figure}[t]
        \centering
        \begin{subfigure}[b]{0.45\textwidth}
            \includegraphics[width=\textwidth]{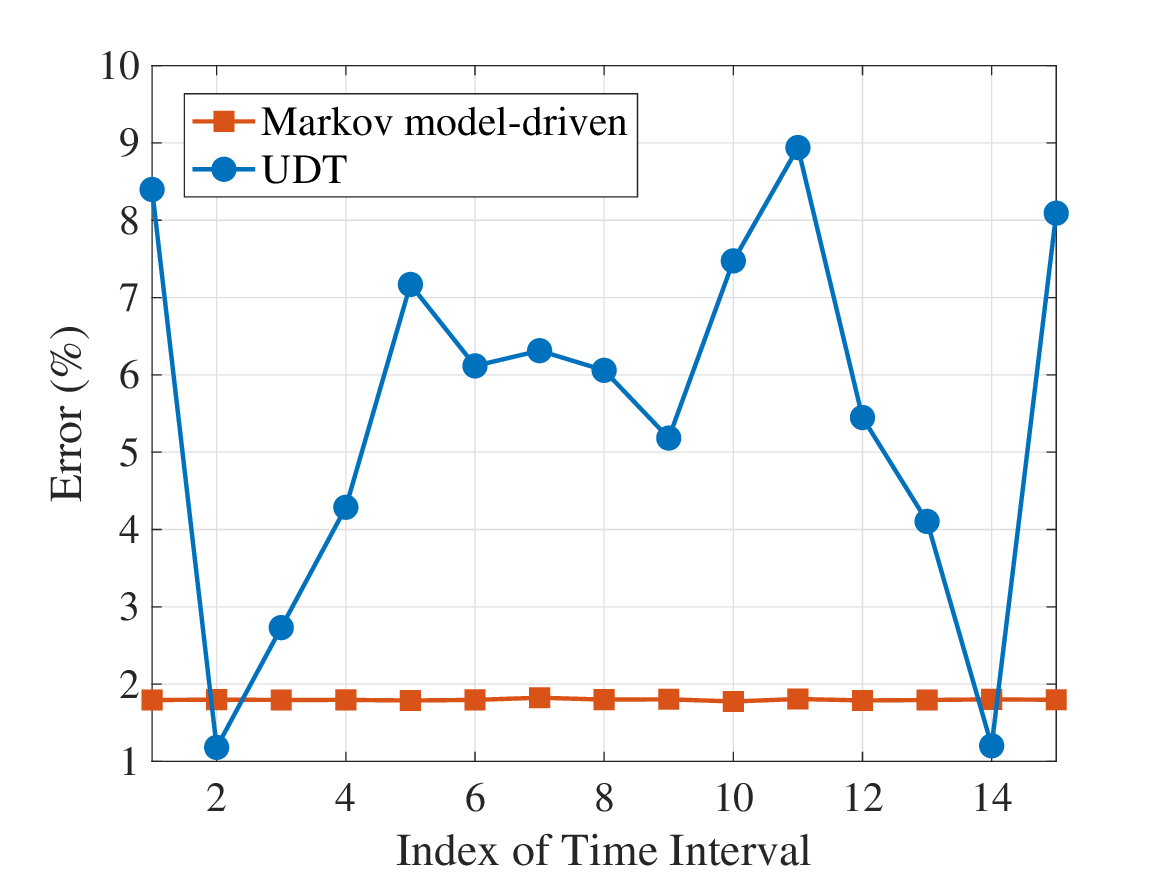}
            \caption{Stationary uplink data rate.}
            \label{fig56:sub1}
        \end{subfigure}
        \quad
        \begin{subfigure}[b]{0.45\textwidth}
            \includegraphics[width=\textwidth]{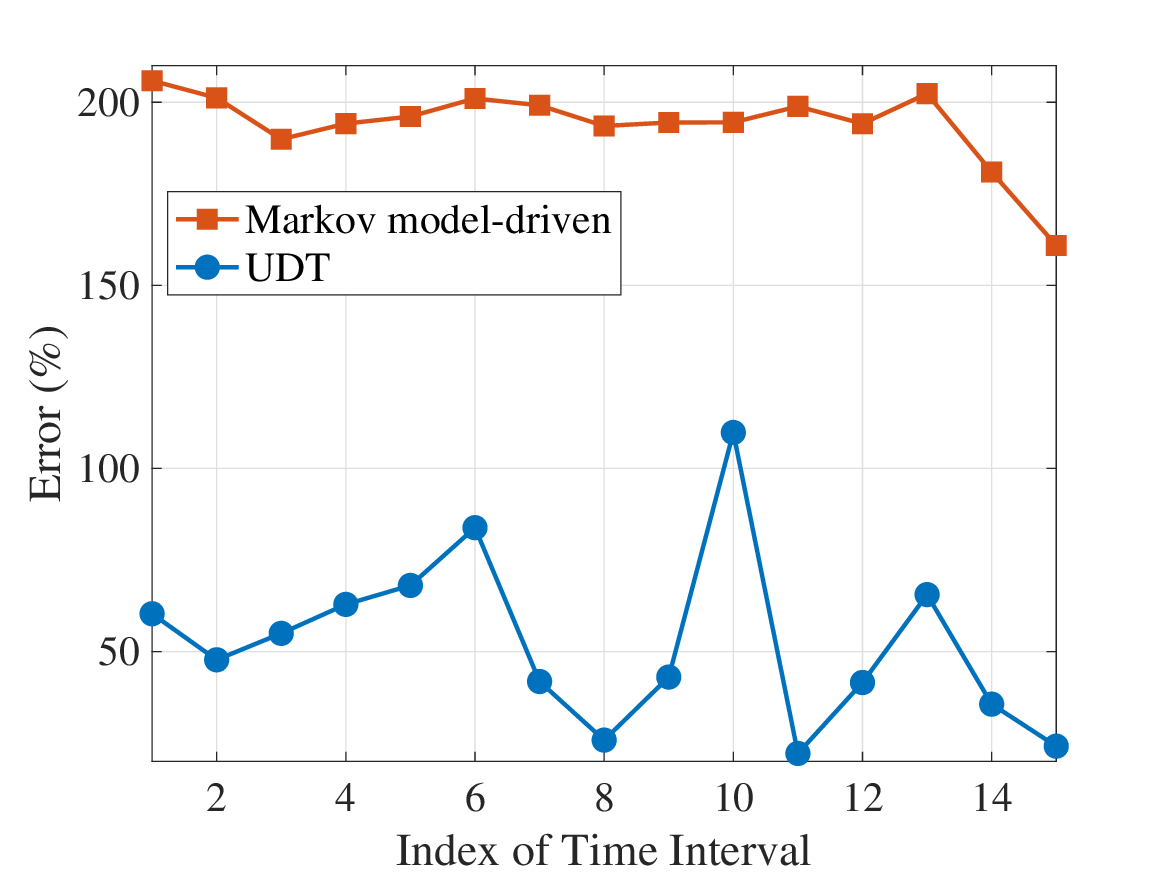}
            \caption{Non-stationary uplink data rate.}
            \label{fig:sub2}
        \end{subfigure}
        \caption{Performance comparison between the UDT and the Markov model-driven approach across 15 time intervals.}
        \label{fig56}
    \end{figure}

In Fig.~\ref{fig56}, we compare the performance of the designed UDT with that of a Markov model-driven approach (labeled as ``Markov model-driven'') across 15 time intervals, in the scenarios with stationary and non-stationery uplink data rates. For this model-driven approach, we pre-define a mathematical model, i.e., $N$~state Markov chain, and estimate its state transition matrix as model parameters according to the statistics of collected data. In~Fig.~\ref{fig56}(a), when the time-varying uplink data rate is stationary across all the time intervals, i.e., the state transition matrix is constant, the model-driven approach slightly outperforms the designed UDT in terms of error. This is because the model-driven approach operates on a known~\emph{a priori} mathematical model rather than using a data model to approximate the mathematical model. However, in~Fig.~\ref{fig56}(b), when the state transition matrix underlying the Markov chain varies across time intervals, the UDT-based approach significantly outperforms the model-driven approach since the designed UDT can capture the time-varying dynamics across time intervals using a data model with DNNs.

\subsection{MBRL for 3D Map Management}

In this subsection, we evaluate the performance of the proposed MBRL scheme for 3D map management using the UDT. 

    \begin{figure}[t]
        \centering
        \includegraphics[width=0.45\textwidth]{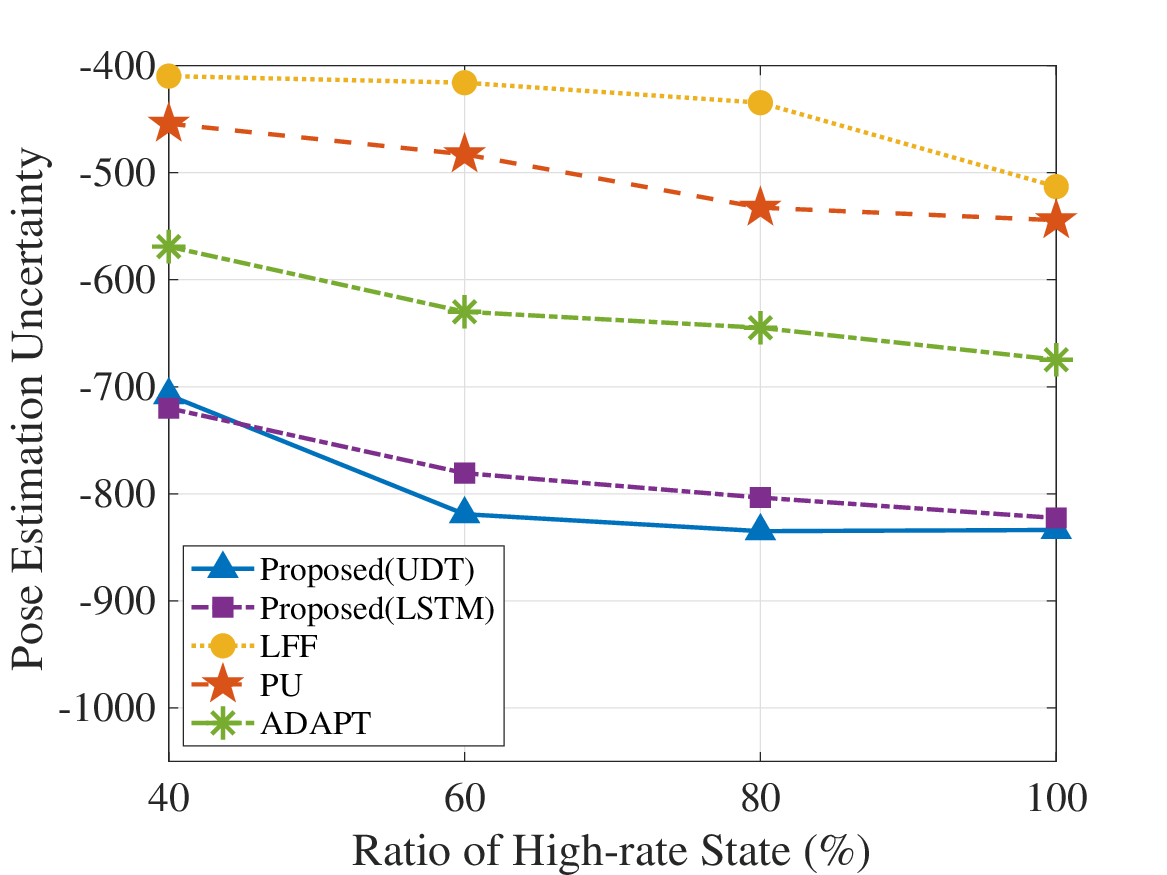}
        \caption{Pose estimation uncertainty versus the ratio of high-rate state when $N=2$.}\label{fig7}
    \end{figure}

In Figs.~\ref{fig7} and~\ref{fig8}, we compare the performance of the proposed MBRL scheme in two cases, one using the UDT (labeled as ``Proposed (UDT)'') and the other using LSTM to generate artificial experiences (labeled as ``Proposed (LSTM)''), with that of the three benchmark 3D map management schemes in one time interval. The uplink data rate follows a two-state Markov chain, with each point representing the average over 15 independent simulation runs. In Fig.~\ref{fig7}, by setting different transition matrices of the two-state Markov chain, we change the ratio of the time slots corresponding to the high-rate state to all time slots. We observe that the proposed MBRL scheme for 3D map management achieves a lower pose estimation uncertainty than the three benchmark schemes. This is because the MBRL scheme, with the help of the UDT, can learn a policy that prioritizes the camera frames for 3D map management by considering their long-term impacts, as opposed to the myopic 3D map management adopted by the three benchmark schemes. This allows the proposed scheme to cope with the dynamics of the uplink data rate and the user's pose. In addition, we can observe that the ``Proposed (UDT)'' and the ``Proposed (LSTM)'' schemes have similar performance when $N=2$, but the former outperforms the latter when $N=4$. This is because the error of LSTM significantly decreases with value of~$N$ as shown in Fig.~\ref{fig2}, thereby reducing the accuracy of the generated artificial experiences.

    \begin{figure}[t]
        \centering
        \includegraphics[width=0.45\textwidth]{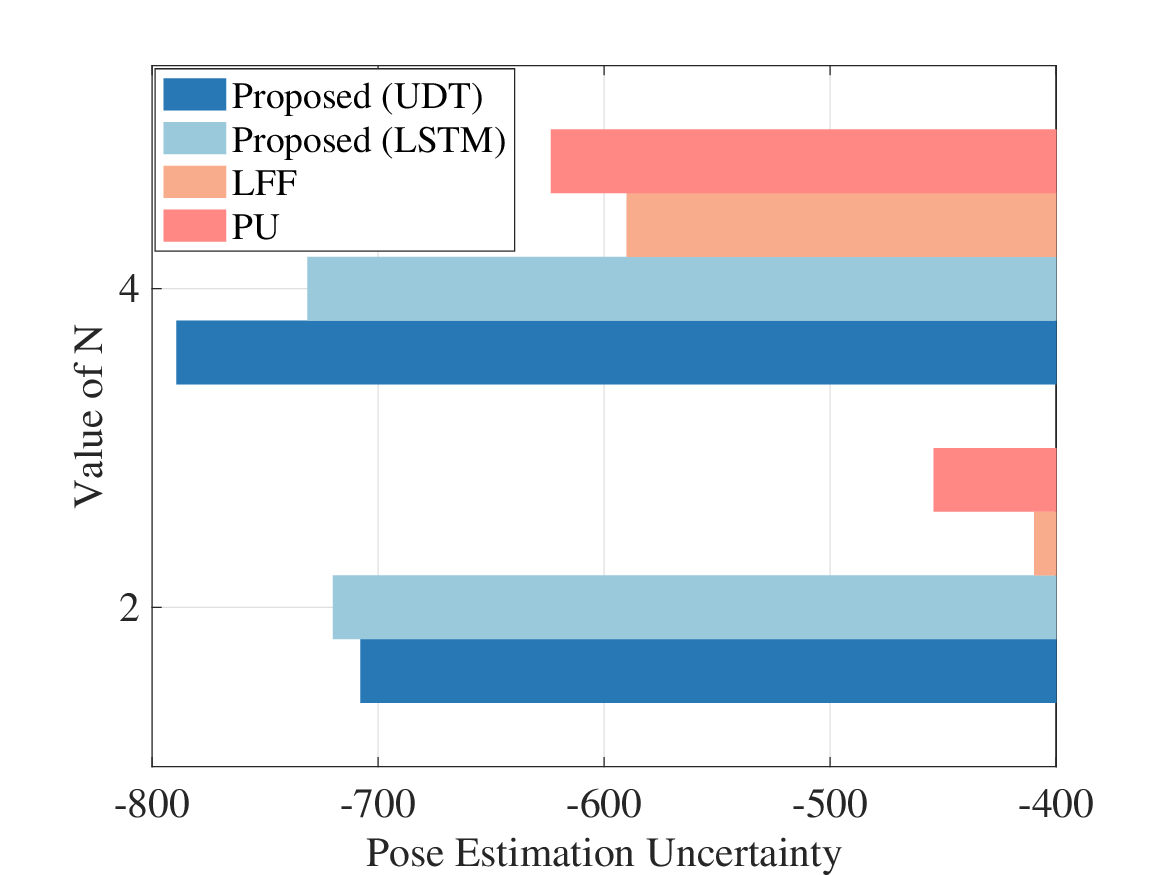}
        \caption{Performance comparison between MBRL and conventional 3D map management schemes when $N=2$ and $N=4$.}\label{fig8}
    \end{figure}

In Fig.~\ref{fig8}, we evaluate the impact of the time-varying uplink data rate on the performance of 3D map management. Specifically, we examine two scenarios, in which the expected uplink data rate is the same while the state transition matrix for the Markov chain has 2 and 4 states, respectively. Compared with the scenario with 4 states, the variance of the uplink data rate across time slots is larger in the scenario with 2 states. We can observe that the performance advantage of the proposed MBRL scheme expands with the value of~$N$. This is because dealing with a large variance in the uplink data rate requires the policy to take into account the long-term impact of 3D map management decision in each time slot on subsequent time slots, which cannot be achieved by the benchmark schemes.

    \begin{figure}[t]
        \centering
        \includegraphics[width=0.45\textwidth]{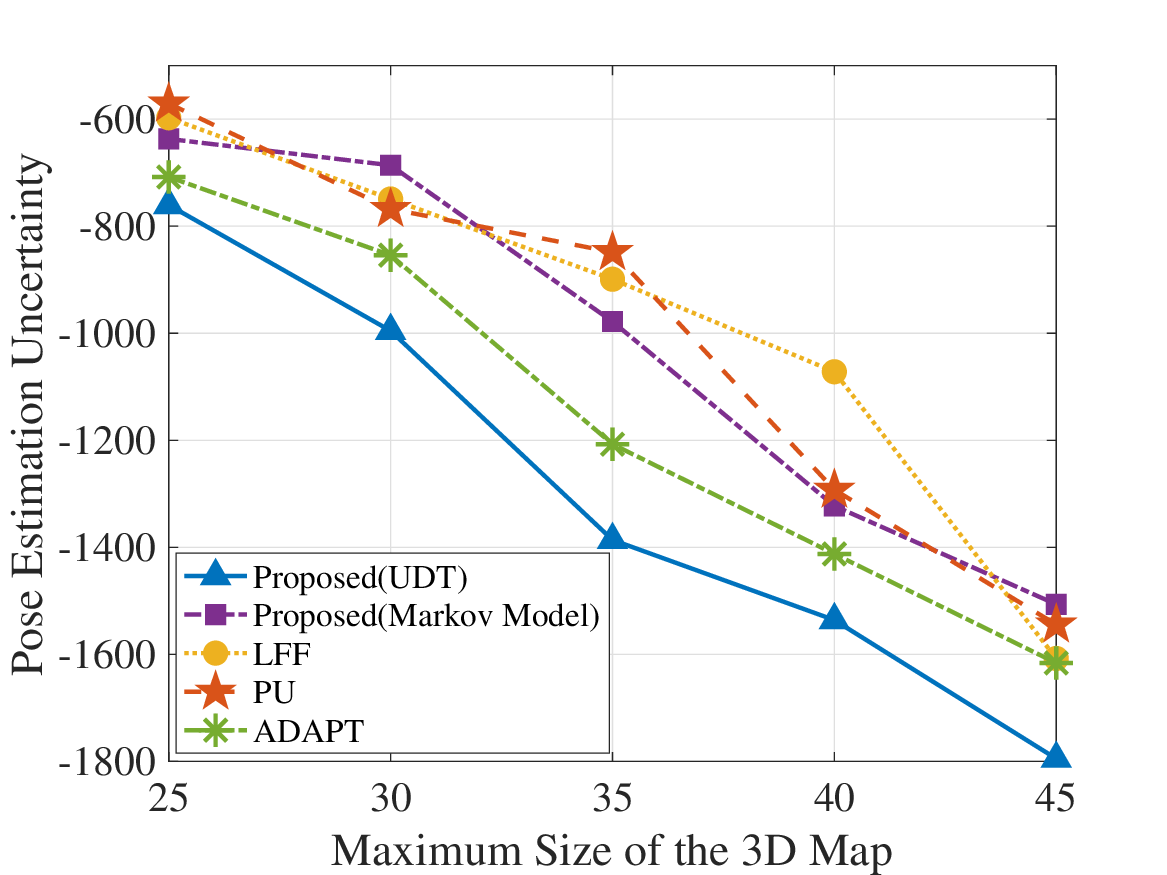}
        \caption{Pose estimation uncertainty versus the 3D map size.}\label{fig9}
    \end{figure}

In Fig.~\ref{fig9}, we compare the performance of the proposed MBRL scheme with that of the three benchmark schemes versus the maximum 3D map size,~$V^\text{max}$, which ranges from $25$ to $45$ camera frames. For each map size, three time intervals are simulated, and the uplink data rate is non-stationary across the three time intervals. In addition to ``LFF'', ``PU'', and ``ADAPT'' schemes, we use the proposed MBRL scheme that employs a Markov model (instead of the UDT) to capture~$x_{t}$ as a benchmark scheme, which is labeled as ``Proposed (Markov model)''. We have two observations from Fig.~\ref{fig9}. First, given different sizes of 3D maps, the proposed MBRL using the UDT scheme can select an appropriate set of camera frames for uploading and updating the 3D map based on their long-term impacts. Thus, it outperforms the ``LFF'', the ``PU'', and the ``ADAPT'' schemes, which make decisions myopically. For example, the ``ADAPT'' scheme solves a myopic uncertainty minimization problem for each time slot rather than a sequential decision-making problem considering the long-term impact of 3D map management decisions on subsequent time slots. Second, the proposed MBRL scheme using the UDT outperforms the ``Proposed (Markov model)'' scheme in terms of pose estimation uncertainty as well. This is because using a fixed Markov model to capture the non-stationary uplink data rate can be inaccurate, as shown in Fig.~\ref{fig56}(b), thereby significantly hampering the capability of the MBRL scheme in learning the optimal 3D map management policy. In contrast, the UDT can cope with the non-stationary uplink data rate and facilitate the proposed MBRL scheme.

\section{Conclusion and Future Work}

In this paper, we have designed a UDT-based 3D map management scheme to facilitate edge-assisted device pose tracking for MAR applications. The UDT established for the MAR device can extract the latent features from the time-varying uplink data rate, thereby supporting the emulation of 3D map management. By using the collected and generated data from the UDT, our MBRL scheme learns a 3D map management policy to prioritizing camera frames for uploading to update the 3D map, which minimizes pose estimation uncertainty. Numerical results have demonstrated the effectiveness of the UDT in capturing the dynamics of the uplink data rate and the adaptivity of the MBRL scheme in coping with the variations in the uplink data rate and the device pose. The designed network dynamics-aware scheme establishes a foundation for customizing UDTs to optimize 3D map management policies based on the distinct network conditions of MAR devices. In the future, we will target efficient resource reservation at an edge server to support 3D map management for multiple MAR devices, considering not only the uplink data rate but also the impacts of different device pose variation patterns on the computing, data storage, and communication resource demands.

\appendix

\subsection{Proof of Lemma~\ref{lemma1}}\label{appendix:lemma1} 
\begin{proof}
    Given a 3D map~$\mathcal{G} = (\mathcal{V}, \mathcal{E})$, the pose estimation uncertainty~$u(\mathcal{G})$ can be calculated according to~\eqref{eq6}, given by~\cite{petersen2008matrix}:
        \begin{equation}\label{eqa1}
            \begin{aligned} 
                 u(\mathcal{G}) & = - \log \left( \det ( \hat{\bm{L}}(\mathcal{G}) \otimes \boldsymbol{\Pi} ) \right)\\
                & = - \log \left( \det ( \hat{\bm{L}}(\mathcal{G}))^{6} \det(\boldsymbol{\Pi})^{|\mathcal{V}|-1} \right),
            \end{aligned} 
        \end{equation}
    where $|\mathcal{V}|$ denotes the number of camera frames in 3D map~$\mathcal{G}$, and the dimension of~$\hat{\bm{L}}(\mathcal{G})$ is $(|\mathcal{V}|-1) \times (|\mathcal{V}|-1)$.

    According to the Kirchhoff's Matrix-Tree Theorem~\cite{godsil2001algebraic}, we can calculate the value of~$\det (\hat{\bm{L}}(\mathcal{G}) )$ for graph~$\mathcal{G}$ based on its weighted number of a spanning tree, i.e.,~the weighted sum of all edges in a tree that connect all nodes in the graph without forming any cycles. Since adding a new edge to a connected graph always increases the weighted number of a spanning tree if the resulting graph remains connected~\cite{khosoussi2014novel}, the following inequality holds:
       \begin{equation}\label{eqa3}
            \kappa(\mathcal{G}) < \kappa(\mathcal{G} \cup \{f\}),\,\, f \notin \mathcal{V},   
       \end{equation}
    where~$f$ denotes a newly added node corresponding to a newly uploaded camera frame, which creates at least one new edge in the 3D map~$\mathcal{G}$. According to the Kirchhoff's Matrix-Tree Theorem and~\eqref{eqa3}, we can derive the following inequality: 
        \begin{equation}\label{eqa4}
            \begin{aligned} 
                 u(\mathcal{G}) & = - \log \left( \kappa(\mathcal{G})^{6} \det(\boldsymbol{\Pi})^{|\mathcal{V}|-1} \right)\\
                                & > - \log \left( \kappa(\mathcal{G} \cup \{f\})^{6} \det(\boldsymbol{\Pi})^{|\mathcal{V}|} \right)\\
                                & = u(\mathcal{G} \cup \{f\}),
            \end{aligned} 
        \end{equation}
    where~$\det(\boldsymbol{\Pi}) \ge 1$ when cameras are high-resolution and high-accuracy and can provide extensive and reliable information for device pose tracking (e.g., an identity matrix is adopted in~\cite{chen2023adaptslam}). Therefore, Lemma~\ref{lemma1} is proved based on~\eqref{eqa4}. 
    \end{proof}


\bibliography{ref}

\bibliographystyle{IEEEtran}

\end{document}